\documentclass[prb,aps,amssymb,twocolumn,notitlepage,floatfix]{revtex4-2}
\usepackage{amsmath}
\usepackage{amssymb}
\usepackage{amsthm}
\usepackage{amsfonts}
\usepackage{listings}
\usepackage{physics}
\usepackage{enumerate}
\usepackage{latexsym}
\usepackage{psfrag}
\usepackage{bm}
\usepackage{graphicx}
\usepackage[caption=false]{subfig}
\usepackage{blkarray}
\usepackage{color}
\usepackage[normalem]{ulem}

\def\selectlanguage#1{}
\usepackage{dsfont}
\usepackage{tabularx}
\usepackage{multirow}
\usepackage{xfp}
\usepackage{calc}
\usepackage[nomessages]{fp}
\usepackage{chemformula}
\usepackage{hyperref}
\usepackage{cleveref}
\usepackage{tikz}
\usepackage{empheq}
\usepackage{capt-of}  

\graphicspath{{images/}}
\NewExpandableDocumentCommand{\bettersquareroot}{O{16}m}{%
  \fpeval{round(sqrt(#2),#1)}%
}
\hypersetup{
     colorlinks = true,
     linkcolor = magenta,
     citecolor = magenta
}

\def\beq{\begin{equation}}
\def\eeq{\end{equation}}
\def\bal{\begin{aligned}}
\def\eal{\end{aligned}}
\def\sg135{$P4_2/mbc1^\prime$ (\# 135)}

\theoremstyle{plain}      
\newtheorem{theorem}{Theorem}

\theoremstyle{definition}
\newtheorem{definition}[theorem]{Definition}

\theoremstyle{remark}

\usepackage{xparse}

\NewDocumentCommand{\cop}{s t! O{k} O{\sigma}}{%
  c\IfBooleanT{#1}{^\dagger}_{\IfBooleanTF{#2}{#3}{\bm{#3}},#4}%
}

\begin{document}

\title{Decomposing momentum scales in the Hubbard Model: From Hatsugai-Kohmoto to Aubry-Andr\'e}

\author{Dmitry Manning-Coe}
\email{dmitry2@illinois.edu}
\affiliation{Department of Physics and Institute for Condensed Matter Theory, University of Illinois at Urbana-Champaign, Urbana IL, 61801-3080, USA}

\author{Barry Bradlyn}
\email{bbradlyn@illinois.edu}
\affiliation{Department of Physics and Institute for Condensed Matter Theory, University of Illinois at Urbana-Champaign, Urbana IL, 61801-3080, USA}

\begin{abstract}
    The all-to-all momentum coupling of the Hubbard interaction makes interacting lattice models generically unsolvable. In many settings, however, from Peierls instabilities to Moir\'e superlattice physics, the low-energy behavior is dominated by scattering at a few characteristic wavevectors. We exploit this by constructing a momentum-space clustering scheme that retains only a chosen subset of interaction channels. Our scheme can be considered a generalization of twist-averaged boundary conditions. In proving this, we also prove that our scheme can be considered as a generalization of Hatsugai-Kohmoto (HK) models, and all versions of the HK model previously considered in the literature arise as special cases. This shows that the surprising phenomenological success of HK models arises from their correspondence to the finite-site Hubbard model. In particular, the recently introduced ``Momentum-Mixing HK'' model corresponds to a specific choice of clustering limit, which is equal to the original finite-site Hubbard model with twist-averaged boundary conditions. Our scheme becomes particularly powerful when a spatially varying potential selects the dominant momentum channels. We demonstrate this on the one-dimensional analogue of interacting moir\'e systems: the Aubry-Andr\'e-Hubbard model. We show that for sufficiently strong onsite potential, clusters as small as two sites can recover the ground state energy to below $1\%$ error relative to DMRG benchmarks. This establishes that physically motivated momentum-space truncations can yield accurate low-energy descriptions at feasible computational cost, opening a path toward tractable interacting models of Moir\'e systems in two dimensions. Code for reproducing all numerical results is available at https://github.com/chainik1125/decomposing-hubbard.
\end{abstract}

\maketitle

\section{Introduction}

The Hubbard model embodies a fundamental tension in many-body
physics: kinetic energy favors delocalized states, while
interactions drive systems towards localization.
For the Hubbard model at intermediate coupling neither tendency is dominant and this competition gives rise to a wealth of correlated phenomena,
from Mott insulators and magnetism~\cite{arovas2022hubbard}, to unconventional
superconductivity~\cite{lee2006doping,keimer2015quantum,balents2020superconductivity}.
It also makes solving the model generically intractable: in momentum space the kinetic energy is diagonal while the interaction couples all momenta, and the resulting Hilbert space grows exponentially with system size.

In many physically important settings, however, macroscopic constraints conspire so that a small number of momentum-transfer channels dominate the low-energy behavior.
This occurs, for example, due to Fermi-surface geometry and nesting in one-dimensional and quasi-one-dimensional systems~\cite{giamarchi2003quantum,shankar1994renormalization,weinberg1986superconductivity,gruner1988dynamics,solyom1979fermi}.
Similar phenomena also occur due to emergent superlattice structure in moir\'e materials, where characteristic reciprocal vectors $\{\mathbf{G}_M\}$ set the scale for both hybridization and scattering~\cite{bistritzer2011bands,koshino2018maximally,han2024quantum}.
This suggests a natural organizing principle: rather than treating all momentum couplings on an equal footing, one can try to construct controlled approximations that retain only the physically dominant channels.

The original Hatsugai-Kohmoto (HK) model~\cite{TheOGHK1992} represents the most extreme version of such a truncation. It discards all inter-momentum coupling components of the interaction. This restores the decomposition of the Hilbert space into independent blocks of well-defined (crystal) momentum $\mathbf{k}$, making the model exactly solvable.
Despite its simplicity, the HK model reproduces a surprising range of Mott physics, including the metal-insulator transition and superconducting instabilities~\cite{2020philipSuperconductor,2022PhilipHKsuperconductor,zhao2023updated}.
The HK interaction, however, is infinitely long-ranged in real space, which obscures the relationship to the original Hubbard model and makes the topological properties of the resulting phases difficult to interpret~\cite{lsm2023,zhao2023failure,ma2025charge,guerci2025electrical,wagner2023mott,setty2024symmetry,setty2024electronic}.

A significant improvement is the ``momentum-mixing HK'' (MMHK) model~\cite{Mai2025momentummixings}, which partitions the Brillouin zone into groups of $n$ maximally separated momenta.
Upon zone folding, these play the role of orbital indices at each point in the reduced zone, yielding an $n$-orbital HK-type Hamiltonian~\cite{lsm2023}.
While this reintroduces some inter-momentum coupling, the MMHK construction does not allow precise control over \textit{which} momentum modes are retained.
This is a potentially significant limitation for moir\'e systems, where the physically relevant momentum transfers are those commensurate with the superlattice vectors $\{\mathbf{G}_M\}$, not those determined by maximal separation in the original Brillouin zone.

In this paper, we introduce a general ``channel-selective'' clustering scheme that partitions the Brillouin zone into interaction clusters of arbitrary size ($N_c$) and separation, retaining only those interaction processes for which all momenta lie within a common cluster.
This construction recovers the original HK model (at $N_c=1$), the MMHK model (for maximal separation), as well as all intermediate cases. It also reduces to the full Hubbard interaction in the limit that $N_c\to N$, the total number of orbitals in the system.

Our central result is that every such cluster-truncated Hamiltonian is unitarily equivalent to an ensemble of independent finite-site Hubbard models with $N_c$ sites and, in general, all-to-all hopping.
In the special case of maximal separation, the hopping matrix simplifies to nearest-neighbor form and the Hamiltonian within each cluster reduces to an $N_c$-site version of the original Hubbard model with twisted boundary condition phase set by the cluster index. 
The full Hamiltonian is then a sum over twist angles, providing an alternative derivation of the twist-averaged boundary condition (TABC) framework~\cite{poilblanc1991twisted,lin2001twist,qin2022hubbard} from a momentum-space perspective.
This equivalence explains the phenomenological success of HK models in reproducing Mott physics: it arises because HK models \textit{are} collections of finite-site Hubbard models.

We then show that non-maximal clustering schemes, which are available from our general construction but not from the MMHK framework, can outperform the maximal scheme in reproducing both the ground state energy and the average occupation of the benchmark DMRG approach. Our clustering scheme converges to the original, thermodynamically large, Hubbard model in the limit when the cluster size is equal to the system size.

The clustering scheme becomes particularly powerful when an external potential selects the dominant momentum channels.
We demonstrate this by analyzing the Aubry-Andr\'e-Hubbard (AAH) model, a one-dimensional analogue of interacting moir\'e systems~\cite{iyer2013mblquasiperiodic,schreiber2015observation}. In this case, we show that increasing the strength of the onsite potential can lead to convergence of our method to the numerically exact DMRG results for a very small number of cluster sites.
We first establish that the incommensurate AAH model is dual to the Aubry-Andr\'e model with HK interactions, connecting the localized and delocalized regimes.
We then show that, in the commensurate case, the onsite potential fuses interaction clusters into ``superclusters''.  In the $t=0$ limit, our clustering scheme exactly recovers the original Hubbard model. For sufficiently strong potential, clusters as small as two sites can recover the ground state energy to below $1\%$ relative error compared to DMRG benchmarks. 
 
A key motivation for this work is the treatment of interactions in moir\'e systems.
Current approaches to twisted bilayer graphene and related materials typically apply a momentum truncation only to the single-particle Hamiltonian, constructing a continuum model within a mini-Brillouin zone, and then add interactions as a separate, post-hoc step~\cite{koshino2018maximally,kwan2025meanfield}. Moreover, only a small number of reciprocal lattice vectors contribute significantly to the projected interaction~\cite{bernevig2021tbgIII}.
Our framework provides a route to treating the kinetic and interaction terms on an equal footing: the same set of retained momentum channels controls both the single-particle hybridization and the interaction processes kept in the truncated model.
Our results on the AAH model serve as a one-dimensional proof of concept for this program.

The structure of the paper is as follows: First, in~\cref{sec:hk_is_hubbard} we introduce a general momentum-space clustering scheme for the Hubbard interaction. We prove that every such scheme is unitarily equivalent to an ensemble of finite-site Hubbard models, providing an alternative derivation of twist-averaged boundary conditions (\cref{thm:HK_is_Hubbard}). We also show that our scheme recovers all previous HK models as special cases, thereby establishing that HK models are simply ensembles of finite site Hubbard models. Finally we show that non-maximal clustering schemes can outperform the accuracy of the MMHK model in recovering the ground state energy and filling at equivalent computational cost.

Next, we turn to the Aubry-Andr\'e-Hubbard model in ~\cref{sec:aah_dual_aahk}. We first show that the AAH model with incommensurate potential is dual to a momentum space ``Aubry-Andr\'e-HK'' (AAHK) model. We then demonstrate that, in the commensurate AAH model, small interaction clusters can quantitatively reproduce DMRG results when the onsite potential is sufficiently strong. We conclude with a discussion of how our results can be applied to interacting moir\'e models and to strongly correlated aperiodic systems more generally.

\section{HK models are finite Hubbard models}\label{sec:hk_is_hubbard}
We first provide a generalized construction for HK models which recovers the original~\cite{TheOGHK1992}, and ``momentum-mixing" \cite{Mai2025momentummixings} HK models as special cases. 
We then use this construction to show that HK models are unitarily equivalent to twist-averaged finite-site Hubbard models with a specific choice of twist angles.

\subsection{Cluster truncated interacting scheme}
\label{sec:general_construction}

We start from the Hubbard model in arbitrary dimension and with any number of hoppings $z$, each indexed by $a$:
\begin{equation}
    \label{eq:original_full_hubbard}
    H=\sum_{\mathbf{i}}\sum_{\sigma}\sum_{a=1}^{z}t^a_{\mathbf{i},\mathbf{i}+\mathbf{r}_a}c^\dagger_{\mathbf{i},\sigma}c_{\mathbf{i+r_a},\sigma}+h.c.+
    U\sum_{\mathbf{i}}n_{\mathbf{i},\uparrow}n_{\mathbf{i},\downarrow},
\end{equation}
with $\mathbf{r}_a$ a Bravais lattice vector denoting the separation for the $a$-th hopping. Transforming into momentum space, we can write the Hamiltonian as
\begin{equation}\label{eq:H_original_momentum}
    H=\sum_{\mathbf{k}}g(\mathbf{k})\cop\cop+
    \frac{U}{N}\sum_{\mathbf{k_1},\mathbf{k_2},\mathbf{q}}\cop*[k_1+q][\uparrow]\cop[k_1][\uparrow]\cop*[k_2-q][\downarrow]\cop[k_2][\downarrow],
\end{equation}
with the hopping factor given by:
\begin{equation}
    g(\mathbf{k})=\sum_{a=1}^{z}t^a e^{i\mathbf{k}\cdot \mathbf{r}_a}+h.c.,\textrm{  }
\end{equation}
where $\mathbf{k}_1,\mathbf{k}_2,$ and $\mathbf{q}$ are in the first Brillouin zone (BZ).

The interaction term in the original Hubbard model couples all momenta in the Brillouin Zone.
To interpolate between the full Hubbard interaction and the HK interaction, we introduce an ``interaction clustering scheme'' which partitions the Brillouin zone into disjoint clusters $\mathcal{C}$ with a fixed number of sites $N_c$, and retains only coupling between momentum modes \textit{within the same cluster} - as shown in \cref{fig:clustering_illustration}.
As we will show below, our construction recovers the MMHK model as the special case corresponding to choosing the momentum-space clusters to be of the maximum possible size but allows for arbitrary momentum modes to be clustered together. Formally, we define an interaction clustering scheme as:

\begin{definition}[Interaction clustering scheme]
\label{def:clustering_scheme}
Let \(\mathcal B\) be the set of momenta of a Brillouin zone in $D$ dimensions.
An \emph{interaction clustering scheme} with cluster size \(N_c\) is specified by the following data:
\begin{enumerate}
    \item A set of \(D\) cluster generator vectors
    \(\bm{\Delta_1},\ldots,\bm{\Delta_D} \in \mathcal B\).

    \item A finite set of within-cluster indices
    $I =\{
        (n^1_1,\ldots,n^1_D),\ldots,(n^{N_c}_1,\ldots,n^{N_c}_D)
    \}$ with $|I| = N_c $.

    \item A set of cluster representatives
    \(\mathcal K \subset \mathcal B\) such that the clusters
    \begin{align}
        \mathcal C_{\mathbf K}
        &=
        \left\{
        \mathbf K + \mathbf k(\mathbf n_j)
        \;:\;
        \mathbf n_j \in I
        \right\},
        &
        \mathbf k(\mathbf n_j)
        &=
        \sum_{d=1}^{D} n^j_d\,\bm{\Delta}_{d},
    \end{align}
    form a disjoint partition of \(\mathcal B\):
    \begin{equation}
        \mathcal B
        = \bigsqcup_{\mathbf K \in \mathcal K}
        \mathcal C_{\mathbf K}.
    \end{equation}
\end{enumerate}
\end{definition}

Given such a scheme, the cluster-truncated Hubbard interaction is obtained by retaining only those quartic terms for which all four momenta lie in a common cluster. 
Note that care must be taken to ensure that the clustering scheme is consistent with the periodic boundary conditions on the BZ. The index parameterization
\begin{equation}
    \mathbf{k}(\mathbf n)=\sum_{d=1}^{D} n_d\,\bm{\Delta}_{d},
    \qquad \mathbf n\in I,\label{eq:cluster_bc}
\end{equation}
does not by itself guarantee closure under addition and subtraction: even if
\(\mathbf{k}\in\mathcal{C}_{\mathbf K}\) and \(\mathbf{q}\in\mathcal{C}_{\mathbf 0}\) are admissible, it may happen that
\(\mathbf{k}\pm\mathbf{q}\notin \mathcal{C}_{\mathbf K}\).
Equivalently, a finite index set \(I\subset \mathbb Z^D\) need not be closed under addition. There are two natural conventions for treating such boundary processes: the \emph{discard} convention and the \emph{wrap} convention, defined as
\begin{enumerate}
    \item \textbf{Discard (open boundary in index space).}
    We set the corresponding matrix elements to zero, i.e. we keep only those terms in which all four momenta
    \(\mathbf{k}_1,\mathbf{k}_2,\mathbf{k}_1+\mathbf{q},\mathbf{k}_2-\mathbf{q}\) lie in \(\mathcal{C}_{\mathbf K}\).
    In this convention, the allowed momentum transfers \(\mathbf{q}\) depend on the location of \(\mathbf{k}_{1}\) and \(\mathbf{k}_{2}\) within the cluster, and the real-space interaction acquires ``edge'' corrections.
    \item  \textbf{Wrap (periodic boundary in index space).}
We treat each cluster as a periodic lattice in index space: 
when addition or subtraction of indices would leave the set $I$, 
we wrap them back using modular arithmetic. In order for this to be possible, we additionally require that \(I\) be a finite abelian group.
In the cases considered in this paper we take
\begin{equation}
    I=\mathbb Z_{M_1}\times\cdots\times \mathbb Z_{M_D},
    \qquad
    N_c=\prod_{d=1}^{D}M_d.
\end{equation}
Concretely, if the index set has $M_d$ values along generator direction $\bm{\Delta}_d$ 
(so that $N_c = \prod_d M_d$), we define
\begin{align}
    \mathbf{k}\oplus\mathbf{q}
    := \sum_{d=1}^{D}\big[(n_d+m_d)\bmod M_d\big]\,\bm{\Delta}_{d},
    \\
    \mathbf{k}\ominus\mathbf{q}
    := \sum_{d=1}^{D}\big[(n_d-m_d)\bmod M_d\big]\,\bm{\Delta}_{d},
    \label{eq:wrap_def_generalD}
\end{align}
whenever $\mathbf{k}=\sum_d n_d\bm{\Delta}_{d}$ 
and $\mathbf{q}=\sum_d m_d\bm{\Delta}_{d}$.
This restores closure under the allowed momentum transfers 
and removes edge effects.
\end{enumerate}
A simple example of the difference between the discard and wrap conventions occurs in one dimension with \(N_c=2\) and \(\Delta=\pi/2a\). 
In this case, \(\mathcal{C}_{K}=\{K,\;K+\pi/2a\}\): choosing \(\mathbf{k} = K+\pi/2a\) and \(\mathbf{q}=\pi/2a\) gives
\(\mathbf{k}+\mathbf{q}=K+\pi/a\notin\mathcal{C}_{K}\).
In the \emph{discard} convention this process is omitted,
whereas in the \emph{wrap} convention it is mapped back into \(\mathcal{C}_{K}\) by the identification
\(\mathbf{k}+\mathbf{q} \mod M=K+\pi/a  \mod \tfrac{\pi}{2a}=K \).

In the main text of this paper \textbf{we always use the wrap} convention. For completeness, in~\cref{sec:appendix_discard_scheme} we treat the minimal separation case in the discard scheme. We also note that we rescale $U/N \to U/N_c$ so that the interaction is normalized over the
retained $N_c$ momentum couplings rather than the original $N$ couplings. Equivalently, after truncation the interaction is no longer averaged over all $N$ momentum
couplings, but only over the retained $N_c$ couplings within each cluster. This is the choice which makes the resulting clustering directly comparable with the original Hubbard model at the same parameters. 

In this case the interaction takes the form:
\begin{equation}
    H^{\mathrm{int}}
    =
    \frac{U}{N_c}\sum_{\mathbf{K}\in\mathcal{K}}
    \sum_{\substack{\mathbf{k}_{1},\mathbf{k}_{2},\mathbf{q}\in \mathcal C_{\mathbf 0}}}
    c^\dagger_{\mathbf{K}+(\mathbf{k}_{1}\oplus\mathbf{q}),\uparrow}c_{\mathbf{K}+\mathbf{k}_{1},\uparrow}
    \,c^\dagger_{\mathbf{K}+(\mathbf{k}_{2}\ominus\mathbf{q}),\downarrow}c_{\mathbf{K}+\mathbf{k}_{2},\downarrow},
    \label{eq:cluster_truncated_interaction}
\end{equation}
where \(\oplus\) and \(\ominus\) denote wrapped addition and subtraction within the cluster.

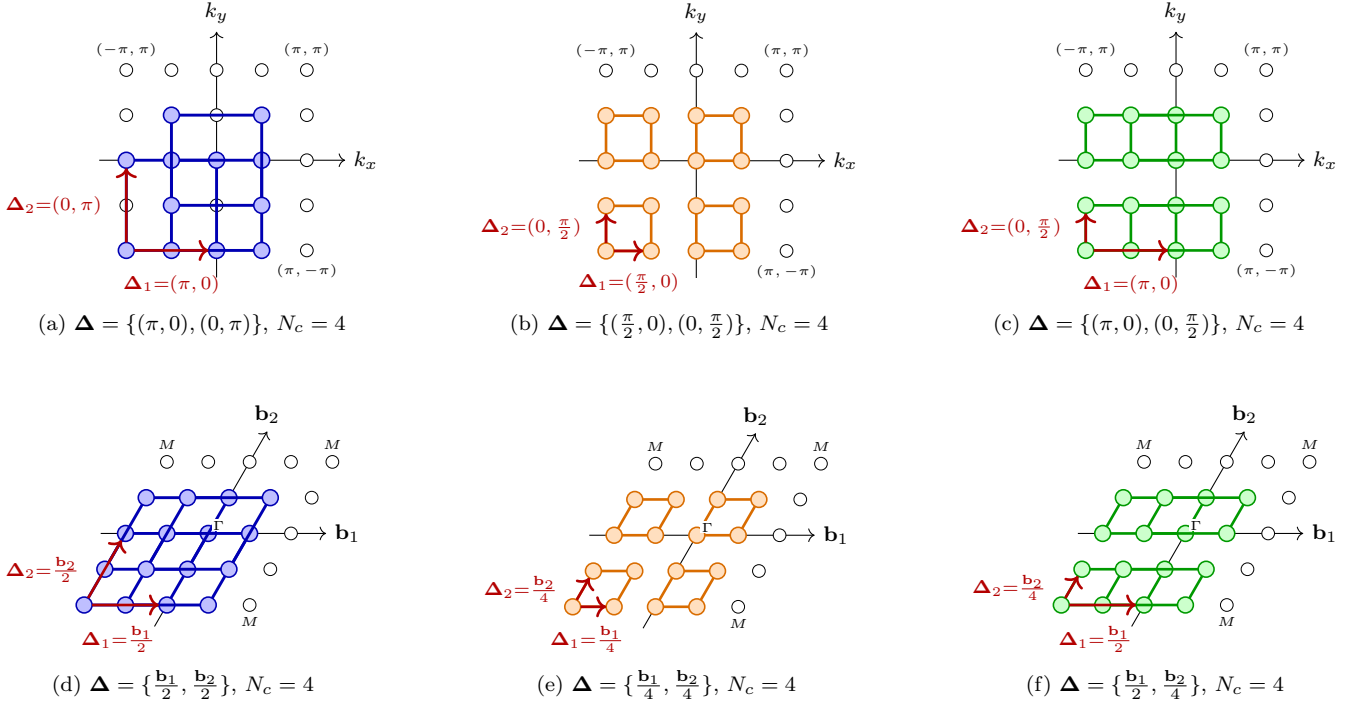
\begin{figure*}[t]
\centering

\subfloat[$\bm{\Delta}=\{(\pi,0),(0,\pi)\}$, $N_c=4$\label{fig:clustering_pi}]{%
\begin{tikzpicture}[x=0.60cm,y=0.60cm,font=\footnotesize]
  \draw[->] (-2.60,0) -- (2.85,0) node[right] {$k_x$};
  \draw[->] (0,-2.60) -- (0,2.85) node[above] {$k_y$};

  \foreach \x in {-2,-1,0,1,2} {
    \foreach \y in {-2,-1,0,1,2} {
      \filldraw[fill=white,draw=black,line width=0.35pt] (\x,\y) circle[radius=2.4pt];
    }
  }

  \foreach \x/\y in {-2/-2, -1/-2, -1/-1} {
    \draw[blue!70!black, line width=1.1pt] (\x,\y) rectangle ++(2,2);
    \foreach \dx/\dy in {0/0,2/0,0/2,2/2} {
      \filldraw[fill=blue!25,draw=blue!70!black,line width=0.6pt]
        (\x+\dx,\y+\dy) circle[radius=3.0pt];
    }
  }

  \draw[->,red!70!black,line width=1.0pt,shorten <=3pt,shorten >=3pt]
    (-2,-2) -- (0,-2);
  \node[below,red!70!black,font=\scriptsize] at (-1,-2.35) {$\bm{\Delta}_1{=}(\pi,0)$};
  \draw[->,red!70!black,line width=1.0pt,shorten <=3pt,shorten >=3pt]
    (-2,-2) -- (-2,0);
  \node[left,red!70!black,font=\scriptsize] at (-2.35,-1) {$\bm{\Delta}_2{=}(0,\pi)$};

  \node[anchor=north,font=\tiny,fill=white,inner sep=1pt,yshift=-4pt] at (2,-2) {$(\pi,-\pi)$};
  \node[anchor=south,font=\tiny,fill=white,inner sep=1pt,yshift=4pt] at (-2,2) {$(-\pi,\pi)$};
  \node[anchor=south,font=\tiny,fill=white,inner sep=1pt,yshift=4pt] at (2,2) {$(\pi,\pi)$};
\end{tikzpicture}%
}%
\hfill
\subfloat[$\bm{\Delta}=\{(\tfrac{\pi}{2},0),(0,\tfrac{\pi}{2})\}$, $N_c=4$\label{fig:clustering_pi2}]{%
\begin{tikzpicture}[x=0.60cm,y=0.60cm,font=\footnotesize]
  \draw[->] (-2.60,0) -- (2.85,0) node[right] {$k_x$};
  \draw[->] (0,-2.60) -- (0,2.85) node[above] {$k_y$};

  \foreach \x in {-2,-1,0,1,2} {
    \foreach \y in {-2,-1,0,1,2} {
      \filldraw[fill=white,draw=black,line width=0.35pt] (\x,\y) circle[radius=2.4pt];
    }
  }

  \foreach \x/\y in {-2/-2, 0/-2, -2/0, 0/0} {
    \draw[orange!85!black, line width=1.0pt] (\x,\y) rectangle ++(1,1);
    \foreach \dx/\dy in {0/0,1/0,0/1,1/1} {
      \filldraw[fill=orange!25,draw=orange!85!black,line width=0.6pt]
        (\x+\dx,\y+\dy) circle[radius=3.0pt];
    }
  }

  \draw[->,red!70!black,line width=1.0pt,shorten <=3pt,shorten >=3pt]
    (-2,-2) -- (-1,-2);
  \node[below,red!70!black,font=\scriptsize] at (-1.5,-2.30) {$\bm{\Delta}_1{=}(\tfrac{\pi}{2},0)$};
  \draw[->,red!70!black,line width=1.0pt,shorten <=3pt,shorten >=3pt]
    (-2,-2) -- (-2,-1);
  \node[left,red!70!black,font=\scriptsize] at (-2.35,-1.5) {$\bm{\Delta}_2{=}(0,\tfrac{\pi}{2})$};

  \node[anchor=north,font=\tiny,fill=white,inner sep=1pt,yshift=-4pt] at (2,-2) {$(\pi,-\pi)$};
  \node[anchor=south,font=\tiny,fill=white,inner sep=1pt,yshift=4pt] at (-2,2) {$(-\pi,\pi)$};
  \node[anchor=south,font=\tiny,fill=white,inner sep=1pt,yshift=4pt] at (2,2) {$(\pi,\pi)$};
\end{tikzpicture}%
}%
\hfill
\subfloat[$\bm{\Delta}=\{(\pi,0),(0,\tfrac{\pi}{2})\}$, $N_c=4$\label{fig:clustering_asym}]{%
\begin{tikzpicture}[x=0.60cm,y=0.60cm,font=\footnotesize]
  \draw[->] (-2.60,0) -- (2.85,0) node[right] {$k_x$};
  \draw[->] (0,-2.60) -- (0,2.85) node[above] {$k_y$};

  \foreach \x in {-2,-1,0,1,2} {
    \foreach \y in {-2,-1,0,1,2} {
      \filldraw[fill=white,draw=black,line width=0.35pt] (\x,\y) circle[radius=2.4pt];
    }
  }

  \foreach \x/\y in {-2/-2, -1/-2, -2/0, -1/0} {
    \draw[green!60!black, line width=1.0pt] (\x,\y) rectangle ++(2,1);
    \foreach \dx/\dy in {0/0,2/0,0/1,2/1} {
      \filldraw[fill=green!20,draw=green!60!black,line width=0.6pt]
        (\x+\dx,\y+\dy) circle[radius=3.0pt];
    }
  }

  \draw[->,red!70!black,line width=1.0pt,shorten <=3pt,shorten >=3pt]
    (-2,-2) -- (0,-2);
  \node[below,red!70!black,font=\scriptsize] at (-1,-2.35) {$\bm{\Delta}_1{=}(\pi,0)$};
  \draw[->,red!70!black,line width=1.0pt,shorten <=3pt,shorten >=3pt]
    (-2,-2) -- (-2,-1);
  \node[left,red!70!black,font=\scriptsize] at (-2.35,-1.5) {$\bm{\Delta}_2{=}(0,\tfrac{\pi}{2})$};

  \node[anchor=north,font=\tiny,fill=white,inner sep=1pt,yshift=-4pt] at (2,-2) {$(\pi,-\pi)$};
  \node[anchor=south,font=\tiny,fill=white,inner sep=1pt,yshift=4pt] at (-2,2) {$(-\pi,\pi)$};
  \node[anchor=south,font=\tiny,fill=white,inner sep=1pt,yshift=4pt] at (2,2) {$(\pi,\pi)$};
\end{tikzpicture}%
}%
\\[0.4cm]

\subfloat[$\bm{\Delta}=\{\tfrac{\mathbf{b}_1}{2},\tfrac{\mathbf{b}_2}{2}\}$, $N_c=4$\label{fig:graphene_clustering_large}]{%
\begin{tikzpicture}[
  x={(0.55cm,0cm)},
  y={(0.275cm,0.476cm)},
  font=\footnotesize
]
  \draw[->] (-2.60,0) -- (2.85,0) node[right] {$\mathbf{b}_1$};
  \draw[->] (0,-2.60) -- (0,2.85) node[above] {$\mathbf{b}_2$};

  \foreach \na in {-2,-1,0,1,2} {
    \foreach \nb in {-2,-1,0,1,2} {
      \filldraw[fill=white,draw=black,line width=0.35pt] (\na,\nb) circle[radius=2.4pt];
    }
  }

  \foreach \na/\nb in {-2/-2, -1/-2, -2/-1, -1/-1} {
    \draw[blue!70!black, line width=1.1pt]
      (\na,\nb) -- ++(2,0) -- ++(0,2) -- ++(-2,0) -- cycle;
    \foreach \da/\db in {0/0,2/0,0/2,2/2} {
      \filldraw[fill=blue!25,draw=blue!70!black,line width=0.6pt]
        (\na+\da,\nb+\db) circle[radius=3.0pt];
    }
  }

  \draw[->,red!70!black,line width=1.0pt,shorten <=3pt,shorten >=3pt]
    (-2,-2) -- (0,-2);
  \node[below,red!70!black,font=\scriptsize] at (-1,-2.35) {$\bm{\Delta}_1{=}\tfrac{\mathbf{b}_1}{2}$};
  \draw[->,red!70!black,line width=1.0pt,shorten <=3pt,shorten >=3pt]
    (-2,-2) -- (-2,0);
  \node[left,red!70!black,font=\scriptsize] at (-2.40,-1) {$\bm{\Delta}_2{=}\tfrac{\mathbf{b}_2}{2}$};

  \node[font=\tiny,fill=white,inner sep=0.5pt,anchor=south west] at (0.05,0.05) {$\Gamma$};
  \node[anchor=south,font=\tiny,fill=white,inner sep=0.5pt,yshift=4pt] at (2,2) {$M$};
  \node[anchor=south,font=\tiny,fill=white,inner sep=0.5pt,yshift=4pt] at (-2,2) {$M$};
  \node[anchor=north,font=\tiny,fill=white,inner sep=0.5pt,yshift=-4pt] at (2,-2) {$M$};

\end{tikzpicture}%
}%
\hfill
\subfloat[$\bm{\Delta}=\{\tfrac{\mathbf{b}_1}{4},\tfrac{\mathbf{b}_2}{4}\}$, $N_c=4$\label{fig:graphene_clustering_small}]{%
\begin{tikzpicture}[
  x={(0.55cm,0cm)},
  y={(0.275cm,0.476cm)},
  font=\footnotesize
]
  \draw[->] (-2.60,0) -- (2.85,0) node[right] {$\mathbf{b}_1$};
  \draw[->] (0,-2.60) -- (0,2.85) node[above] {$\mathbf{b}_2$};

  \foreach \na in {-2,-1,0,1,2} {
    \foreach \nb in {-2,-1,0,1,2} {
      \filldraw[fill=white,draw=black,line width=0.35pt] (\na,\nb) circle[radius=2.4pt];
    }
  }

  \foreach \na/\nb in {-2/-2, 0/-2, -2/0, 0/0} {
    \draw[orange!85!black, line width=1.0pt]
      (\na,\nb) -- ++(1,0) -- ++(0,1) -- ++(-1,0) -- cycle;
    \foreach \da/\db in {0/0,1/0,0/1,1/1} {
      \filldraw[fill=orange!25,draw=orange!85!black,line width=0.6pt]
        (\na+\da,\nb+\db) circle[radius=3.0pt];
    }
  }

  \draw[->,red!70!black,line width=1.0pt,shorten <=3pt,shorten >=3pt]
    (-2,-2) -- (-1,-2);
  \node[below,red!70!black,font=\scriptsize] at (-1.5,-2.30) {$\bm{\Delta}_1{=}\tfrac{\mathbf{b}_1}{4}$};
  \draw[->,red!70!black,line width=1.0pt,shorten <=3pt,shorten >=3pt]
    (-2,-2) -- (-2,-1);
  \node[left,red!70!black,font=\scriptsize] at (-2.40,-1.5) {$\bm{\Delta}_2{=}\tfrac{\mathbf{b}_2}{4}$};

  \node[font=\tiny,fill=white,inner sep=0.5pt,anchor=south west] at (0.05,0.05) {$\Gamma$};
  \node[anchor=south,font=\tiny,fill=white,inner sep=0.5pt,yshift=4pt] at (2,2) {$M$};
  \node[anchor=south,font=\tiny,fill=white,inner sep=0.5pt,yshift=4pt] at (-2,2) {$M$};
  \node[anchor=north,font=\tiny,fill=white,inner sep=0.5pt,yshift=-4pt] at (2,-2) {$M$};

\end{tikzpicture}%
}%
\hfill
\subfloat[$\bm{\Delta}=\{\tfrac{\mathbf{b}_1}{2},\tfrac{\mathbf{b}_2}{4}\}$, $N_c=4$\label{fig:graphene_clustering_asym}]{%
\begin{tikzpicture}[
  x={(0.55cm,0cm)},
  y={(0.275cm,0.476cm)},
  font=\footnotesize
]
  \draw[->] (-2.60,0) -- (2.85,0) node[right] {$\mathbf{b}_1$};
  \draw[->] (0,-2.60) -- (0,2.85) node[above] {$\mathbf{b}_2$};

  \foreach \na in {-2,-1,0,1,2} {
    \foreach \nb in {-2,-1,0,1,2} {
      \filldraw[fill=white,draw=black,line width=0.35pt] (\na,\nb) circle[radius=2.4pt];
    }
  }

  \foreach \na/\nb in {-2/-2, -1/-2, -2/0, -1/0} {
    \draw[green!60!black, line width=1.0pt]
      (\na,\nb) -- ++(2,0) -- ++(0,1) -- ++(-2,0) -- cycle;
    \foreach \da/\db in {0/0,2/0,0/1,2/1} {
      \filldraw[fill=green!20,draw=green!60!black,line width=0.6pt]
        (\na+\da,\nb+\db) circle[radius=3.0pt];
    }
  }

  \draw[->,red!70!black,line width=1.0pt,shorten <=3pt,shorten >=3pt]
    (-2,-2) -- (0,-2);
  \node[below,red!70!black,font=\scriptsize] at (-1,-2.35) {$\bm{\Delta}_1{=}\tfrac{\mathbf{b}_1}{2}$};
  \draw[->,red!70!black,line width=1.0pt,shorten <=3pt,shorten >=3pt]
    (-2,-2) -- (-2,-1);
  \node[left,red!70!black,font=\scriptsize] at (-2.40,-1.5) {$\bm{\Delta}_2{=}\tfrac{\mathbf{b}_2}{4}$};

  \node[font=\tiny,fill=white,inner sep=0.5pt,anchor=south west] at (0.05,0.05) {$\Gamma$};
  \node[anchor=south,font=\tiny,fill=white,inner sep=0.5pt,yshift=4pt] at (2,2) {$M$};
  \node[anchor=south,font=\tiny,fill=white,inner sep=0.5pt,yshift=4pt] at (-2,2) {$M$};
  \node[anchor=north,font=\tiny,fill=white,inner sep=0.5pt,yshift=-4pt] at (2,-2) {$M$};

\end{tikzpicture}%
}%
\caption{Illustration of momentum-space clustering on a square lattice (top) and a triangular lattice as in graphene (bottom). Red arrows indicate the generator vectors $\bm{\Delta}_1$, $\bm{\Delta}_2$ connecting momenta within a cluster. Top row: \(\pi/2\)-spaced \(5\times5\) square grid. (a)~$\bm{\Delta}=\{(\pi,0),(0,\pi)\}$: each cluster (blue) is a $\pi\times\pi$ square. (b)~$\bm{\Delta}=\{(\frac{\pi}{2},0),(0,\frac{\pi}{2})\}$: each cluster (orange) is a $\frac{\pi}{2}\times\frac{\pi}{2}$ square. (c)~Asymmetric $\bm{\Delta}=\{(\pi,0),(0,\frac{\pi}{2})\}$: clusters (green) are $\pi\times\frac{\pi}{2}$ rectangles, illustrating independent control of the spacing in each direction. Bottom row: analogous clustering on the reciprocal lattice of the triangular Bravais lattice (as in graphene), displayed in the parallelogram BZ along $\mathbf{b}_1$, $\mathbf{b}_2$ (at $60^\circ$). Clusters are parallelogram-shaped; $\Gamma$ is labeled at the origin. The three corner points of the parallelogram BZ are all equivalent to the $M$ point of the hexagonal BZ, related by reciprocal lattice vectors. The honeycomb sublattice enters as an orbital index.}
\label{fig:clustering_illustration}
\end{figure*}

The total Hamiltonian can thus be decomposed as a sum over cluster Hamiltonians $H_{\mathbf{K}}$ as
\begin{align}
    H_\mathbf{\Delta}&=\sum_{\mathbf{K}}^{N/N_c}H_{\mathbf{K}} \nonumber \\
    H_{\mathbf{K}}&=\sum_{k_j=1}^{N_c}
    g(\mathbf{K}+\mathbf{k_j})\cop*[K+k_j]\cop[K+k_j] \nonumber \\ 
    +&\frac{U}{N_c}
    \sum_{\mathbf{k}_{1},\mathbf{k}_{2},\mathbf{q}\in \mathcal{C}_{\mathbf{0}}}
    c^\dagger_{\mathbf{K}+(\mathbf{k}_{1}\oplus\mathbf{q}),\uparrow}c_{\mathbf{K}+\mathbf{k}_{1},\uparrow}
    \,c^\dagger_{\mathbf{K}+(\mathbf{k}_{2}\ominus\mathbf{q}),\downarrow}c_{\mathbf{K}+\mathbf{k}_{2},\downarrow}.
\end{align}

For example, in the scheme shown in \cref{fig:clustering_pi} $\mathbf{k}_j=n^j_1(0,\pi)+n^j_2(\pi,0)$ and $(n^j_1,n^j_2)\in\{(0,0),(0,1),(1,0),(1,1)\}$.

We now prove the following theorem:

\begin{theorem}
\label{thm:HK_is_Hubbard}
For any clustering scheme defined by Definition~\ref{def:clustering_scheme} using the wrapped condition, there exists a unitary transformation which maps the clustered Hamiltonian to a sum over decoupled cluster Hamiltonians $H_{\mathbf{K}}$,
\[
H_{\bm\Delta}=\sum_{\mathbf K\in\mathcal K} H_{\mathbf K},
\]
where each $H_{\mathbf{K}}$ is a finite size Hubbard Hamiltonian with \(N_c\) sites and of the form:
\[
H_{\mathbf K}
=
\sum_{a=1}^{z} t_a e^{i\mathbf K\cdot \mathbf r_a}
\sum_{\alpha,\beta,\sigma}
J^a_{\alpha\beta}\,
c^{\dagger,\mathbf K}_{\alpha,\sigma}c^{\mathbf K}_{\beta,\sigma}
+\mathrm{h.c.}
+
U\sum_{\alpha=1}^{N_c}
n^{\mathbf K}_{\alpha,\uparrow}n^{\mathbf K}_{\alpha,\downarrow},
\]
with \(J^a_{\alpha\beta}\) a potentially all-to-all hopping matrix determined by the clustering scheme.
\end{theorem}

We prove this by introducing the discrete Fourier transform on the finite
abelian group $I$ and acting with it on each term. This will allow us to give an explicit form for the hopping kernel $J^a_{\alpha\beta}$ in \cref{eq:hopping_kernel}. Our proof will proceed in two parts by focusing on the interaction term in \cref{sec:interaction_proof}, and the hopping term in \cref{sec:hopping_proof}.

\subsection{Interaction term}\label{sec:interaction_proof}

We start with the interaction term. The key observation is that the wrap
convention turns each cluster into the finite abelian group
$I=\prod_{d=1}^{D}\mathbb Z_{M_d}$. Since the Hubbard interaction strength $U$
is independent of momentum, the interaction depends only on this group
structure. The natural basis change is therefore the discrete Fourier transform
on $I$.

Let $\mathbf n_j=(n^j_1,\ldots,n^j_D)\in I$ label the relative cluster momenta
and $\mathbf m_\alpha=(m^\alpha_1,\ldots,m^\alpha_D)\in I$ label the sites of
the $N_c$-site auxiliary Hubbard model. We define
\begin{equation}
    \mathbf{k}_j=\mathbf{k}(\mathbf n_j)=\sum_{d=1}^{D}n^j_d\,\bm{\Delta}_{d},
\end{equation}
and the finite-model reciprocal and direct lattice vectors
\begin{equation}
\label{eq:def_k0}
\mathbf{k}^0_j=\sum_{d=1}^{D}\frac{2\pi n^j_d}{M_d a}\hat{\mathbf g}_d,
\qquad
\mathbf{R}^0_\alpha=a\sum_{d=1}^{D}m^\alpha_d\,\hat{\mathbf e}_d,
\end{equation}
where $\hat{\mathbf e}_d$ are the primitive lattice vectors of the original Bravais lattice, $\hat{\mathbf g}_d$ are the corresponding reciprocal primitive vectors, and $\hat{\mathbf g}_d\cdot\hat{\mathbf e}_{d'}=\delta_{dd'}$.

We then define the unitary transforms:
\begin{equation}
\label{eq:alpha_transform}
c_{\mathbf K+\mathbf k_j,\sigma}
=\frac{1}{\sqrt{N_c}}\sum_{\alpha=1}^{N_c} e^{+i\mathbf k^0_j\cdot\mathbf R^0_\alpha}\,c^{\mathbf K}_{\alpha,\sigma},
\end{equation}
\begin{equation}
c^{\mathbf K}_{\alpha,\sigma}
=\frac{1}{\sqrt{N_c}}\sum_{j=1}^{N_c} e^{-i\mathbf k^0_j\cdot\mathbf R^0_\alpha}\,c_{\mathbf K+\mathbf k_j,\sigma},
\end{equation}
where the orthogonality relation factorizes over the product group:
\begin{equation}
\frac{1}{N_c}\sum_{j=1}^{N_c}e^{i\mathbf k^0_j\cdot(\mathbf R^0_\alpha-\mathbf R^0_\beta)}
=
\prod_{d=1}^{D}
\left[
\frac{1}{M_d}\sum_{n_d=0}^{M_d-1}
e^{2\pi i n_d (m^\alpha_d-m^\beta_d)/M_d}
\right]
=\delta_{\alpha\beta}.
\end{equation}
Because both $\mathbf{k}_j$ and $\mathbf{k}^0_j$ are parameterized by the same
multi-index $\mathbf n_j$, wrapped addition in the interaction cluster is
mapped to addition of finite-model momenta modulo the corresponding reciprocal
lattice vectors. Since the interaction strength $U$ is independent of the momenta,
the cluster interaction depends only on this additive structure and not on the
numerical values of the momenta themselves.

Now let $\mathbf K$ index an interaction cluster, and let
$\mathcal C=\{\mathbf{k}(\mathbf n):\mathbf n\in I\}$ be the corresponding set
of relative momenta equipped with the wrapped addition inherited from $I$.
Then we can prove that the clustered Hubbard interaction,
\begin{equation}
\label{eq:truncated_int_term}
\begin{split}
H^{(\mathbf K)}_{\rm int}
&=\frac{U}{N_c}\sum_{\mathbf k_1,\mathbf k_2,\mathbf q\in\mathcal C}
c^\dagger_{\mathbf K+(\mathbf k_1\oplus\mathbf q),\uparrow}c_{\mathbf K+\mathbf k_1,\uparrow}\\
&\qquad\times\;c^\dagger_{\mathbf K+(\mathbf k_2\ominus\mathbf q),\downarrow}c_{\mathbf K+\mathbf k_2,\downarrow},
\end{split}
\end{equation}
becomes purely diagonal (``onsite'') in the transformed basis:
\begin{equation}
    \label{eq:alpha_basis_interaction}
H^{(\mathbf K)}_{\rm int}=U\sum_{\alpha=1}^{N_c} n^{\mathbf K}_{\alpha,\uparrow}n^{\mathbf K}_{\alpha,\downarrow}.
\end{equation}

\begin{proof}
Introduce a ``coarse cluster density operator" (CCDO):
\begin{equation}
\label{eq:coarse_cluster_density_operator}
\rho^{(\mathbf K)}_\sigma(\mathbf q)=\sum_{\mathbf k\in\mathcal C}
c^\dagger_{\mathbf K+(\mathbf k\oplus\mathbf q),\sigma}c_{\mathbf K+\mathbf k,\sigma}.
\end{equation}
Expressed in terms of these density operators the truncated cluster Hamiltonian \cref{eq:truncated_int_term} becomes:
\begin{equation}
H^{(\mathbf K)}_{\rm int}=\frac{U}{N_c}\sum_{\mathbf q\in\mathcal C}\rho^{(\mathbf K)}_\uparrow(\mathbf q)\rho^{(\mathbf K)}_\downarrow(-\mathbf q).
\end{equation}  
Using the unitary transform and closure of $\mathcal C$, if
$\mathbf q=\mathbf{k}(\mathbf m_q)$ then
$\mathbf q^0=\sum_{d=1}^{D}\tfrac{2\pi (m_q)_d}{M_d a}\hat{\mathbf g}_d$. Summing over $\mathbf k$ forces the density operator $\rho^{(\mathbf K)}_\sigma(\mathbf q)$ to be diagonal in $\alpha,\beta$:
\begin{align}
\label{eq:density_op}
\rho^{(\mathbf K)}_\sigma(\mathbf q)
&=\frac{1}{N_c}\sum_{\mathbf k\in\mathcal C}\sum_{\alpha,\beta}
e^{-i(\mathbf k^0+\mathbf q^0)\cdot\mathbf R^0_\alpha}
e^{+i\mathbf k^0\cdot\mathbf R^0_\beta}
c^{\dagger,\mathbf K}_{\alpha,\sigma}c^{\mathbf K}_{\beta,\sigma}
\nonumber\\
&=\sum_{\alpha,\beta}\left[
\frac{1}{N_c}\sum_{\mathbf k\in\mathcal C}
e^{i\mathbf k^0\cdot(\mathbf R^0_\beta-\mathbf R^0_\alpha)}
\right]
e^{-i\mathbf q^0\cdot\mathbf R^0_\alpha}
c^{\dagger,\mathbf K}_{\alpha,\sigma}c^{\mathbf K}_{\beta,\sigma}
\nonumber\\
&=\sum_\alpha e^{-i\mathbf q^0\cdot\mathbf R^0_\alpha}n^{\mathbf K}_{\alpha,\sigma},
\end{align}
where in the first line we used that $(\mathbf k\oplus\mathbf q)^0$ differs from $\mathbf k^0+\mathbf q^0$ only by a reciprocal lattice vector of the finite auxiliary lattice. Explicitly, the extra phase introduced by the wrap is trivial on $\mathbf R^0_\alpha$, since for any auxiliary reciprocal lattice vector $\mathbf G_{\rm aux}$, $e^{-i\mathbf G_{\rm aux}\cdot \mathbf R^0_\alpha}=1$. In the second line we used the discrete completeness relation
\[
\frac{1}{N_c}\sum_{\mathbf k\in\mathcal C}
e^{i\mathbf k^0\cdot(\mathbf R^0_\beta-\mathbf R^0_\alpha)}
=\delta_{\alpha\beta}.
\]
Thus the $\mathbf k$-sum makes $\rho^{(\mathbf K)}_\sigma(\mathbf q)$ diagonal in the $\alpha$ basis. The remaining sum over $\mathbf q$ in \cref{eq:diagonal_alpha_interaction} uses the same completeness relation once more, now to force the two site indices in the product of densities to coincide.
Explicitly, substituting the density operators \cref{eq:density_op} into the truncated cluster Hamiltonian \cref{eq:truncated_int_term}:
\begin{align}
\label{eq:diagonal_alpha_interaction}
H^{(\mathbf K)}_{\rm int}
&=\frac{U}{N_c}\sum_{\alpha,\beta}\Big(\sum_{\mathbf q\in\mathcal C}e^{i\mathbf q^0\cdot(\mathbf R^0_\beta-\mathbf R^0_\alpha)}\Big)
n^{\mathbf K}_{\alpha,\uparrow}n^{\mathbf K}_{\beta,\downarrow}
\nonumber \\
&=U\sum_\alpha n^{\mathbf K}_{\alpha,\uparrow}n^{\mathbf K}_{\alpha,\downarrow},
\end{align}
where the last equality uses discrete completeness $\sum_{\mathbf q\in\mathcal C}e^{i\mathbf q^0\cdot(\mathbf R^0_\beta-\mathbf R^0_\alpha)}=N_c\delta_{\alpha\beta}$.
\end{proof}

We note that this argument applies for any separation - in particular it does \emph{not} require maximal separation. It applies to any
wrapped clustering scheme for which the index set carries the finite abelian
group structure \(I=\prod_{d=1}^{D}\mathbb Z_{M_d}\); the only further
requirement is that $U$ is constant across the retained momentum channels. \footnote{%
Equivalently, the transform \cref{eq:alpha_transform} is the discrete Fourier
transform on the finite abelian group \(I\). The on-site Hubbard interaction is
the real-space form of a momentum-independent two-body coupling on any such
group, and hence the two are related by Fourier transform regardless of the
specific spacing vectors.}

\subsection{Real-space form of the cluster interaction.}
We can build a physical intuition for the truncated interaction \cref{eq:diagonal_alpha_interaction} by considering its real-space form in a particularly transparent limit. We consider ``maximal'' separations - those where the momentum-space separation is as large as it can be for a given cluster size. In Appendix~\ref{sec:app_minimal_sep_real_space}, we consider the opposite limit of the ``minimal'' separation case, which groups together nearest-neighbor sites in momentum space. 

To lighten notation, we consider the $1$D case. To extend to higher-dimensional Bravais lattices we can replace the scalar decomposition \(R=(XN_c+\alpha)a\) with the componentwise decomposition \(\mathbf R=\mathbf R_{\mathbf X}+\mathbf R^0_\alpha\), where \(\mathbf R_{\mathbf X}=a\sum_d X_d M_d\,\hat{\mathbf e}_d\) and \(\mathbf R^0_\alpha=a\sum_d m^\alpha_d\,\hat{\mathbf e}_d\). For maximal separations \(\bm{\Delta}^{\rm max}_d=\mathbf G_d/M_d\), the phase sums then factorize over \(d\).
The maximal separation case is defined by setting the separation to:
\begin{equation}
    \Delta_{\textrm{max}}=\frac{2\pi}{L}\frac{N}{N_c}=\frac{2\pi}{N_c a}
\end{equation}
In this case we can replace $k\oplus q$ with regular addition modulo the
\textit{lattice} reciprocal vector $G_{\rm lat}=\tfrac{2\pi}{a}$.

Intuitively, our clustering in reciprocal space coarse-grains the underlying Hubbard lattice at the length scales given by the separation vectors $\Delta$. As we go to larger cluster sizes, we are able to resolve more of the on-site Hubbard interaction.
To make this transparent we parameterize the microscopic lattice vectors by coordinates $X$ for the coarse lattice, and $\alpha$ for the sites within it, such that $R=(XN_c+\alpha)a$. We then transform \Cref{eq:truncated_int_term} into real space under this parameterization:

\begin{align}
    \label{eq:real_space_maximal_intermediate}
    H_{\mathrm{int}}&=
    \frac{U}{N_c N^2}
    \sum_{X_{1},...,X_{4}}\sum_{\alpha_1,...,\alpha_4}\sum_{K}\sum_{k_1,k_2,q\in\mathcal{C}_K}\nonumber\\
    &e^{i K(X_1+X_3-X_2-X_4)N_c a} 
    e^{i K(\alpha_1+\alpha_3-\alpha_2-\alpha_4)a} \nonumber\\
    &\times e^{i k_1(X_1-X_2)N_c a}
    e^{i k_2(X_3-X_4)N_c a}
    e^{i q(X_1-X_3)N_c a} \nonumber \\
    &\times e^{i k_1(\alpha_1-\alpha_2)a}
    e^{i k_2(\alpha_3-\alpha_4)a}
    e^{i q(\alpha_1-\alpha_3) a}\nonumber\\
    &\times c^\dagger_{X_1N_c+\alpha_1,\uparrow}c_{X_2N_c+\alpha_2,\uparrow}
    c^\dagger_{X_3N_c+\alpha_3,\downarrow}c_{X_4N_c+\alpha_4,\downarrow}.
\end{align}
We note in this expression that we were able to pull apart $q$ from $k_1$ and
$k_2$ only because we are in the maximal separation scheme so that
$k_1\oplus q=k_1+q-G_{\rm lat}$ when wrap occurs.

We then perform the sums in turn. For the inner sum over $k_1,k_2,q$ we note that:
\begin{equation}
    k_m (X N_c a)=\frac{2\pi}{N_c} m X N_c a
    =2\pi m X,
\end{equation}
and hence these exponentials are trivial. The remaining terms involving $\alpha$ are exactly reciprocal to the cluster momenta $k_1,k_2,q$ and so:
\begin{align}
    &\sum_{\alpha_1,...,\alpha_4}\sum_{k_1,k_2,q\in\mathcal{C}_K}
    e^{i k_1(\alpha_1-\alpha_2)a}
    e^{i k_2(\alpha_3-\alpha_4)a}
    e^{i q(\alpha_1-\alpha_3) a} \nonumber\\
    \label{eq:alpha_deltas}
    &=N_c^3\sum_{\alpha_1,...,\alpha_4}
    \delta_{\alpha_1,\alpha_2}\delta_{\alpha_3,\alpha_4}\delta_{\alpha_1,\alpha_3},
\end{align}
and so only the operators diagonal in $\alpha$ remain.

To perform the outer sum over $K$ in \cref{eq:real_space_maximal_intermediate}, we note that in the maximal separation scheme, the set of cluster representatives $K$ is given by the set of 
$N_X=N/N_c$ smallest crystal momenta,  $K\in\{0,\tfrac{2\pi}{L},2\tfrac{2\pi}{L},...,(N_X-1)\tfrac{2\pi}{L} \}$. Denoting $X_1+X_3-X_2-X_4=\Delta X$, we write the sum over $K$ as:
\begin{align}
    \sum_{m=0}^{N_X-1}e^{i \tfrac{2\pi}{L}m\Delta X N_c a}
    =\sum_{m=0}^{N_X-1}e^{i \tfrac{2\pi}{N_X}m\Delta X}=N_X\delta_{\Delta X, 0 \mod N_X},
\end{align}
where the last equality follows from the fact that the coarse lattice coordinate $X$ and the cluster representative $m$ index a real and reciprocal lattice respectively of size $N_X$, and are hence dual to each other. Since \cref{eq:alpha_deltas} enforces $\alpha_1=\alpha_2=\alpha_3=\alpha_4$, we have
\[
R_1+R_3-R_2-R_4
=
N_c a\,(X_1+X_3-X_2-X_4)
=
N_c a\,\Delta X.
\]
Hence the condition $\Delta X=0 \pmod{N_X}$ is equivalent to
\[
R_1+R_3-R_2-R_4=0 \pmod{L},
\qquad
L=N_XN_c a.
\]

The resulting real space interaction is then:
\begin{align}
    \label{eq:real_space_maximal_int}
    &H_{\mathrm{int}}=\frac{U}{N_X}\sum_{X_1,...,X_4}
    \sum_{\alpha}\delta_{X_1+X_3,X_2+X_4}
    c^\dagger_{X_1N_c+\alpha,\uparrow}c_{X_2N_c+\alpha,\uparrow} \nonumber\\
    &\times c^\dagger_{X_3N_c+\alpha,\downarrow}c_{X_4N_c+\alpha,\downarrow},
\end{align}
where we have rewritten the coarse-lattice constraint \(\delta_{\Delta X,\,0\mod N_X}\), using the relation above, as conservation of the physical center of mass modulo the full system size \(L\).

Eq.~\eqref{eq:real_space_maximal_int} has a direct physical interpretation. The interaction partitions the underlying lattice into a coarse "super-lattice" with sites separated by $N_c a$, and the sites within the superlattice unit cell indexed by $\alpha$. \textit{Within} a cluster, the interaction is the regular local Hubbard interaction. \textit{Between} clusters, the interaction is HK. As we increase the number of sites in a cluster $N_c$, we thus interpolate between the HK and Hubbard interactions.

\subsection{Hopping term}\label{sec:hopping_proof}
We now move on to evaluate the hopping term under the transformation \cref{eq:alpha_transform}. Rewriting the set of $z$ hoppings defined by \cref{eq:H_original_momentum} in terms of our clustering scheme yields
\begin{equation}
    T_{\mathbf K}
    =
    \sum_{j=1}^{N_c}\sum_{\sigma}
    g(\mathbf K+\mathbf k_j)\,
    c^\dagger_{\mathbf K+\mathbf k_j,\sigma}
    c_{\mathbf K+\mathbf k_j,\sigma}.
\end{equation}
Inserting \cref{eq:alpha_transform} gives
\begin{equation}
    T_{\mathbf K}
    =
    \frac{1}{N_c}\sum_{j=1}^{N_c}\sum_{\alpha,\beta=1}^{N_c}\sum_{\sigma}
    g(\mathbf K+\mathbf k_j)\,
    e^{i\mathbf k^0_j\cdot(\mathbf R^0_\alpha-\mathbf R^0_\beta)}
    c^{\dagger,\mathbf K}_{\alpha,\sigma}c^{\mathbf K}_{\beta,\sigma}.
\end{equation}
Using $g(\mathbf k)=\sum_{a=1}^{z}t_a e^{i\mathbf k\cdot\mathbf r_a}+\mathrm{h.c.}$,
we obtain
\begin{align}
    T_{\mathbf K}
    &=
    \sum_{a=1}^{z}t_a e^{i\mathbf K\cdot\mathbf r_a}
    \sum_{\alpha,\beta=1}^{N_c}\sum_{\sigma}
    J^a_{\alpha\beta}\,
    c^{\dagger,\mathbf K}_{\alpha,\sigma}c^{\mathbf K}_{\beta,\sigma}
    +\mathrm{h.c.},
\end{align}
where the hopping matrix elements $J^a_{\alpha\beta}$ are defined for each hopping vector $\mathbf r_a$ as follows. First, we define the projection of $\mathbf{r}_a$ onto the clustering vectors
\begin{equation}
    \eta_d^a:=\frac{M_d}{2\pi}\,\bm{\Delta}_{d}\cdot\mathbf r_a,
    \qquad
    \boldsymbol\eta^a:=a\sum_{d=1}^{D}\eta_d^a\hat{\mathbf e}_d.
\end{equation}
Then
\begin{equation}
    e^{i\mathbf k_j\cdot \mathbf r_a}
    =
    \exp\!\left(i\sum_{d=1}^{D}\frac{2\pi n_d^j}{M_d}\eta_d^a\right)
    =
    e^{i\mathbf k^0_j\cdot\boldsymbol\eta_a},
\end{equation}
and so
\begin{align}
    \label{eq:hopping_kernel}
    J^a_{\alpha\beta}
    &=
    \frac{1}{N_c}\sum_{j=1}^{N_c}
    e^{i\mathbf k_j^0\cdot(\mathbf R^0_\alpha-\mathbf R^0_\beta+\boldsymbol\eta_a)}.
\end{align}
Combining Eq.~\eqref{eq:hopping_kernel} with the results of Sec.~\ref{sec:interaction_proof} yields for the full clustered Hamiltonian
\begin{align}
    \label{eq:full_cluster_ham}
    H=\sum_{\mathbf K\in\mathcal K}&\left[
    \sum_{a=1}^{z}t_a e^{i\mathbf{K}\cdot\mathbf{r}_a}
    \sum_{\alpha,\beta=1}^{N_c}\sum_{\sigma}
    J^{a}_{\alpha\beta}c^{\dagger,\mathbf K}_{\alpha,\sigma}c^{\mathbf K}_{\beta,\sigma}
    +\mathrm{h.c.}\right.\nonumber \\
    &\left.+U\sum_{\alpha=1}^{N_c}n^{\mathbf K}_{\alpha,\uparrow}n^{\mathbf K}_{\alpha,\downarrow}
    \right].
\end{align}
This proves \cref{thm:HK_is_Hubbard}.
The term in square brackets is the cluster Hamiltonian $H_\mathbf{K}$, and is the Hamiltonian for a finite-site Hubbard model with $N_c$ sites with hopping determined by $J^a_{\alpha\beta}$.
Each hopping in the original Hubbard model $t_a$ in general results in a hopping between all pairs of sites $J^a_{\alpha\beta}$. Each hopping also acquires a phase factor $e^{i\mathbf{K}\cdot\mathbf{r}_a}$ determined by the reference cluster momentum $\mathbf{K}$ and the hopping distance $\mathbf{r}_a$. 
\subsection{Momentum mixing HK $=$ finite-site Hubbard with twist averaging}
\label{sec:philip_scheme}

We now consider the special case in which the cluster spacing is maximal along
each generator direction:
\begin{equation}
    \mathbf G_d:=\frac{2\pi}{a}\hat{\mathbf g}_d,
    \qquad
    \bm{\Delta}^{\rm max}_{d}
    =
    \frac{1}{M_d}\mathbf G_d
    =
    \frac{2\pi}{M_d a}\hat{\mathbf g}_d.
\end{equation}
We prove that in this case the Hamiltonian \cref{eq:full_cluster_ham}
recovers the ``momentum-mixing'' HK model. For example, in the case considered
in~\cite{Mai2025momentummixings} of a four-site clustering $N_c=4$ on a square
lattice, one has $M_x=M_y=2$ and hence the maximal separations are
$\Delta_{1}=(\pi/a,0)$ and $\Delta_2=(0,\pi/a)$. The corresponding relative
momenta are $\mathbf{k}_1=(0,0)$, $\mathbf{k}_2=\mathbf{\Delta}_1$,
$\mathbf{k}_3=\mathbf{\Delta}_2$, and
$\mathbf{k}_4=\mathbf{\Delta}_1+\mathbf{\Delta}_2$, with the transformed
operators given by inserting these values into the inverse of
\cref{eq:alpha_transform}:
\begin{align}
    c^{\dagger,\mathbf{K}}_{\alpha=1}
    &=\frac{1}{2}\left[c^\dagger_{\mathbf{K}}
    +c^\dagger_{\mathbf{K}+\mathbf{k}_2}+c^\dagger_{\mathbf{K}+\mathbf{k}_3}+c^\dagger_{\mathbf{K}+\mathbf{k}_4}\right],\\
    c^{\dagger,\mathbf{K}}_{\alpha=2}
    &=\frac{1}{2}\left[c^\dagger_{\mathbf{K}}
    -c^\dagger_{\mathbf{K}+\mathbf{k}_2}+c^\dagger_{\mathbf{K}+\mathbf{k}_3}-c^\dagger_{\mathbf{K}+\mathbf{k}_4}\right],\\
    c^{\dagger,\mathbf{K}}_{\alpha=3}
    &=\frac{1}{2}\left[c^\dagger_{\mathbf{K}}
    +c^\dagger_{\mathbf{K}+\mathbf{k}_2}-c^\dagger_{\mathbf{K}+\mathbf{k}_3}-c^\dagger_{\mathbf{K}+\mathbf{k}_4}\right],\\
    c^{\dagger,\mathbf{K}}_{\alpha=4}
    &=\frac{1}{2}\left[c^\dagger_{\mathbf{K}}
    -c^\dagger_{\mathbf{K}+\mathbf{k}_2}+c^\dagger_{\mathbf{K}+\mathbf{k}_3}+c^\dagger_{\mathbf{K}+\mathbf{k}_4}\right].
\end{align}
which correspond to those given in the supplement of~\cite{mmhk_arxiv}. 
In general, for maximal separation the relative momentum $k_j$ is given by:
\begin{equation}
    \mathbf{k}_j=\sum_{d=1}^D n^j_{d}\bm{\Delta}^{\rm max}_{d}
    =\sum_{d=1}^D n^j_{d}\tfrac{2\pi}{M_d a}\hat{\mathbf g}_d
    =\mathbf{k}^0_j,
\end{equation}
where the last equality is just the definition of $\mathbf{k}^0_j$ in \cref{eq:def_k0}. This allows us to simplify the cluster hopping matrix $J_{\alpha\beta}$ considerably:
\begin{equation}
    J^a_{\alpha\beta}
    =
    \frac{1}{N_c}\sum_{j=1}^{N_c}
    e^{i\mathbf{k}^0_j\cdot(\mathbf{R}^0_{\alpha}-\mathbf{R}^0_{\beta}+\mathbf{r}_a)}
    =
    \delta_{\mathbf m_\alpha,\mathbf m_\beta+\mathbf s_a},
\end{equation}
where $\mathbf r_a=a\sum_{d=1}^{D}s_{a,d}\hat{\mathbf e}_d$ and the addition of
multi-indices is understood modulo $I$.
The Hamiltonian \cref{eq:full_cluster_ham} hence becomes
\begin{align}
    \label{eq:mmhk_ham}
 H=\sum_{\mathbf K\in\mathcal K}&\left[
    \sum_{a=1}^{z}t_a e^{i\mathbf{K}\cdot\mathbf{r}_a}
\sum_{\alpha=1}^{N_c}\sum_{\sigma}\left(
    c^{\dagger,\mathbf K}_{\alpha,\sigma}c^{\mathbf K}_{\alpha+\mathbf s_a,\sigma}
+\mathrm{h.c.}\right)\right.\nonumber \\
&\left.+U\sum_{\alpha=1}^{N_c}n^{\mathbf K}_{\alpha,\uparrow}n^{\mathbf K}_{\alpha,\downarrow}
    \right],
\end{align}
Here $\alpha+\mathbf s_a$ denotes the site whose multi-index is
$\mathbf m_\alpha+\mathbf s_a$ modulo $I$. The cluster Hamiltonian
$H_{\mathbf{K}}$ in square brackets is the same Hamiltonian as the original Hubbard
Hamiltonian \cref{eq:original_full_hubbard}, but for $N_c$ rather than $L$
sites and with a phase factor $e^{i\mathbf{K}\cdot\mathbf{r}_a}$ multiplying each hopping $t^a$. This is also the Hamiltonian of the Momentum-Mixing HK model. We provide an explicit mapping to the case considered in~\cite{Mai2025momentummixings} in \cref{sec:mmhk_notation_comparison} to ease comparison.
This provides an alternative derivation of the correspondence first noted in Ref.~\cite{bai2025proofmomentummixinghatsugai}, which established the equivalence between the MMHK model and the twist-averaged Hubbard model in the thermodynamic limit. Here we show that the mapping is exact for any finite cluster size $N_c$.

We can also better understand this Hamiltonian by noting that equation~\cref{eq:full_cluster_ham} already shows that the cluster label $\mathbf K$ appears as a Peierls phase in the hopping. For a general clustering scheme this gives a finite-site Hubbard model with generalized all-to-all hoppings $J^a_{\alpha\beta}$. In the maximal-separation case, these hoppings reduce to the ordinary finite-lattice ones, and for $\mathbf r_a=a\sum_{d=1}^{D}s_{a,d}\hat{\mathbf e}_d$ the phase may be written as
\[
e^{i\mathbf K\cdot\mathbf r_a}
=
e^{i\sum_{d=1}^{D}\theta_{\mathbf K,d}\,s_{a,d}},
\qquad
\theta_{\mathbf K,d}:=a\,\mathbf K\cdot\hat{\mathbf e}_d.
\]
Thus, in our coordinate conventions, the corresponding twist vector is $\boldsymbol\theta_{\mathbf K}=a\mathbf K$.

This establishes that the maximal spacing HK construction with $N_c$ sites is unitarily equivalent to a twist-averaged  $N_c$-site Hubbard Hamiltonian with the twist angles $\{\boldsymbol\theta_{\mathbf K}\}=\{a\mathbf K\}$ corresponding to the chosen cluster representatives $\mathbf K$. We emphasize that this equivalence holds for any finite cluster size, and does not require the thermodynamic limit. This equivalence has two important consequences.

First, this provides an alternative derivation of conventional twist averaging from a momentum-space perspective. Surprisingly, this derivation shows that twist-averaging is equivalent to truncating the momentum modes in the interacting term. This allows us to interpret the generalized HK models that come from our construction within the twist-averaging framework. Since the maximal-separation scheme reproduces the standard twist-averaged boundary-condition construction, the non-maximal schemes can be viewed as a generalization of twist averaging.

For any wrapped clustering, the cluster label $\mathbf K$ enters as a Peierls phase, so the Hamiltonian decomposes into a sum of twist sectors labelled by $\mathbf K$. Maximal separation is special because the relative cluster momenta of the twist sector $\mathbf{k}$ coincide with the crystal momenta of the original lattice. At maximal separation, a physical hopping by lattice vector $\mathbf r_a=a\sum_d s_{a,d}\hat{\mathbf e}_d$ maps under the transformation \cref{eq:alpha_transform} to an integer translation by $\mathbf s_a$ on the auxiliary $N_c$-site lattice, so the kinetic term remains local. For non-maximal separations the cluster momenta no longer coincide with the auxiliary lattice momenta. The discrete Fourier transform on the cluster group still localizes the interaction, but the microscopic hoppings now correspond to \emph{fractional} translations on the auxiliary cluster. Because a fractional translation is not local on the lattice, these hoppings do not localize: they instead appear as the long-ranged all-to-all kernel $J^a_{\alpha\beta}$.

We also note an important computational advantage of the transform \cref{eq:alpha_transform}. For each interaction cluster, it replaces the \(O(N_c^3)\) momentum-space couplings in \cref{eq:cluster_truncated_interaction} by the \(N_c\) onsite Hubbard terms in \cref{eq:alpha_basis_interaction}, at the cost of introducing a dense \(N_c\times N_c\) one-body hopping matrix. In other words, the transform trades a complicated two-body interaction for a local interaction plus a more complicated kinetic term, which is simpler to handle.

Second, we note that this equivalence clarifies what is at first sight a puzzling observation about HK models.  Given that HK models are manifestly nonlocal, it is at first sight puzzling how they are able to recover the strongly \emph{localized} physics of the Hubbard model. In particular, HK models have been shown to exhibit the metal-insulator transition~\cite{TheOGHK1992}, dynamic spectral weight transfer~\cite{tenkila_dynamical_spectral_weight_2025}, and diverging self-energies~\cite{lsm2023,setty2024symmetry}. The equivalence in Eq.~\eqref{eq:mmhk_ham} reconciles this tension in a straightforward way: the observed correspondence between (non-local) HK and (local) Hubbard models arises simply because HK model are equivalent to twist-averaged, finite-site Hubbard models. Moreover, the previously-mentioned phenomena for which the HK model is successful are among those that can be captured in exact diagonalization studies of small Hubbard clusters~\cite{dagotto1994correlated}.

\subsection{Alternative clustering schemes}

This equivalence of the cluster truncated Hamiltonian with the original finite-site Hubbard model only holds for the case in which the interaction clustering spacings $\{\mathbf{\Delta_d}\}$ are maximal. Our general construction, however, allows us to retain arbitrary modes. In the general case, the hopping $\mathbf{r}_{a}$\ leads to a fractional contribution to the hopping matrix $J_{\alpha\beta}$ in \cref{eq:hopping_kernel}. In this case, in general all couplings are non-zero and there is no choice of twist angles for which the ensemble in \cref{eq:full_cluster_ham} is equivalent to a single cluster Hamiltonian with twist averaging.

\begin{figure*}[tbp]
    \centering
    \includegraphics[width=0.98\linewidth]{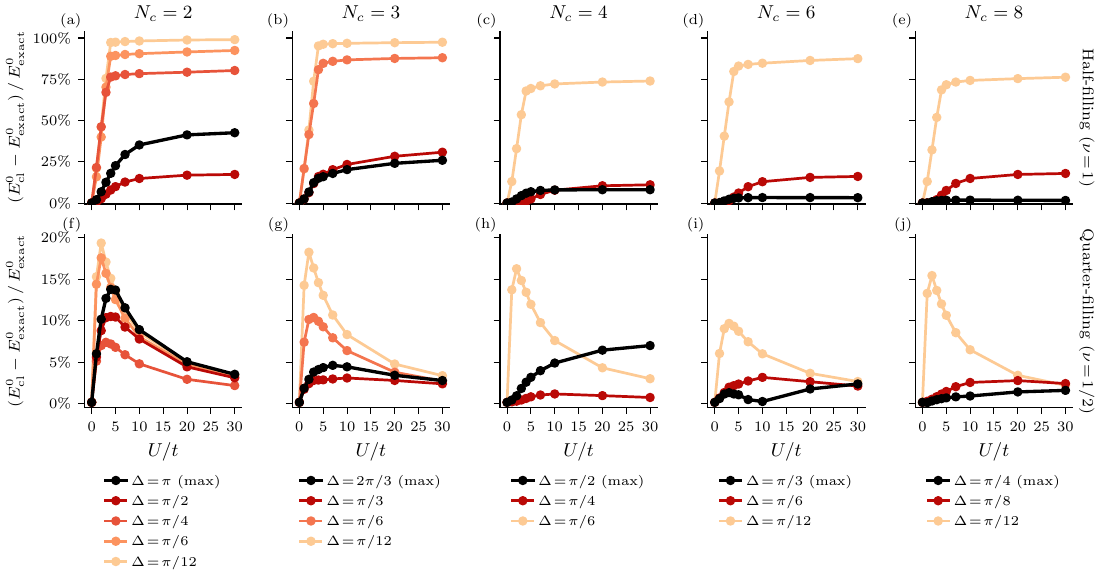}
    \caption{Relative error in the ground state energy per site of the
    one-dimensional Hubbard model ($\lambda=0$) for different interaction
    clustering schemes, benchmarked against the exact Bethe ansatz solution for a finite, periodic chain of $L=48$ sites.
    Each column corresponds to a fixed interaction cluster size $N_c$
    (from $N_c=2$ to $N_c=8$); the top row shows half-filling ($n=1$)
    and the bottom row quarter-filling ($n=0.5$). Within each panel,
    different curves correspond to different cluster spacings $\Delta$:
    the maximal separation (momentum-mixing HK, black) is compared
    against non-maximal alternatives. At half-filling the different
    schemes perform comparably, but at quarter-filling non-maximal
    spacings can dramatically outperform the maximal scheme---for
    example, at $N_c=4$ the $\Delta=\pi/4$ scheme recovers the exact
    energy to within numerical precision across most $U/t$ values,
    while the maximal $\Delta=\pi/2$ scheme incurs $5$--$8\%$ relative
    error. Convergence is also not monotonic in $N_c$: at
    quarter-filling the $N_c=3$ scheme outperforms $N_c=4$.}
    \label{fig:scheme_comparison_v0}
\end{figure*}

Physically, non-maximal clusterings can be interpreted as selecting an \emph{arbitrary}, but finite set of momentum-transfer channels in the interaction: by construction we keep only terms for which $(\mathbf{k_1},\mathbf{k}_2,\mathbf{k}_1+\mathbf{q},\mathbf{k}_2-\mathbf{q})$ remain within the subgroup generated by $\{\mathbf{\Delta}_d\}$. 
This mirrors the organizing principle of moir\'e systems. A moir\'e superlattice introduces long-wavelength Fourier components at small reciprocal vectors $\{\mathbf G_{{\rm M},\ell}\}$, which strongly hybridize electronic states with $\mathbf{k}$ and $\mathbf{k}+\mathbf G_{{\rm M},\ell}$ and produce an emergent mini-Brillouin zone.
In that setting, the dominant low-energy scattering processes are those with momentum transfers built from the moir\'e vectors.
Our non-maximal schemes allow one to impose this structure directly at the level of the interaction: choosing $\{\mathbf{\Delta}_d\}$ to match the physically relevant transfers (e.g. $2k_F$ in 1D Peierls physics, or moir\'e reciprocal vectors in 2D) yields reduced Hubbard clusters tailored to the emergent supercell rather than to the original microscopic lattice.

\subsection{Numerical comparison of different clustering schemes to the Hubbard model}
We now present a numerical comparison of how different clustering schemes \cref{eq:full_cluster_ham} compare to the original Hubbard model \cref{eq:original_full_hubbard} in one dimension.
In \cref{fig:scheme_comparison_v0} we show the ground state energy at half- (top row) and quarter-filling (bottom row) for each clustering scheme, along with the ground state energy of the Hubbard model obtained through the Density Matrix Renormalization Group (DMRG) for a one-dimensional chain much longer than the cluster size.

For the purposes of this comparison, it is important to match both the boundary conditions and the finite system size of the reference calculation to those of the clustering scheme. At $\lambda=0$, where the underlying model reduces to the one-dimensional Hubbard chain, we therefore implement the finite-size periodic Bethe-ansatz/Takahashi solution on an $L=48$ ring as the reference~\cite{liebwu1968hubbard,essler2005hubbard,essler2013shellfilling}, which we verify recovers the exact diagonalization ground state energies up to machine precision for Hubbard rings up to size 12. The full code for all figures is available at: https://github.com/chainik1125/decomposing-hubbard .
 
In each case, we see that there is a generalized scheme which outperforms (in terms of relative deviation from the DMRG ground state energy per site) the maximal separation scheme for a wide range of $U$ values, especially at quarter-filling. 
At quarter-filling we can see in the bottom row of \cref{fig:scheme_comparison_v0} that non-maximal schemes can outperform across a range of $U$ values. The improvement is particularly striking at quarter-filling for $N_c=4$. Across most $U$-values, the non-maximal scheme at $\Delta=\pi/4$ separation recovers the DMRG energy up to numerical error, whereas the maximal scheme $\Delta=\pi/2$ has a relative error between $5$-$8\%$ higher.

We also note that the convergence to DMRG is generally not monotonic in the cluster size, especially at quarter-filling. Thus we see that the $N_c=3$ scheme in \cref{fig:scheme_comparison_v0} (e) significantly outperforms the $N_c=4$ scheme (f).

In Appendix~\ref{sec:app_numerical_results} we provide additional data on the comparison of the \textit{fillings} obtained from DMRG to the different clustering schemes considered here as we vary the chemical potential $\mu_0$ and the interaction $U$. The resulting comparison is largely consistent with what we see for the ground state energy; across most parameter ranges there is a non-maximal scheme which outperforms the maximal one.

\section{The Aubry-Andr\'e Hubbard and Aubry-Andr\'e HK models}

The general construction outlined in \cref{sec:general_construction} provides a way to retain specific momentum channels in the interaction. 
This scheme becomes particularly important for understanding the low-temperature behavior of systems where certain momentum couplings dominate the low-energy physics.
This situation occurs in moir\'e-like systems, where a spatially varying supercell potential couples nearby momentum modes to create an emergent supercell.

In the rest of this work, we study the simplest system which retains this physics, namely the one-dimensional interacting Aubry-Andr\'e model~\cite{harper1955single,aubry1980analyticity,iyer2013mblquasiperiodic}. The Hamiltonian for the Aubry-Andr\'e Hubbard (AAH) model is given by
\begin{align}
    \label{eq:real_space_aah}
    H&=t\sum_{i=1}^{N}\sum_{\sigma}\left(c^\dagger_{i+1,\sigma}c_{i,\sigma}+h.c.\right) \nonumber \\ 
    &+\lambda\sum_{i=1}^{N}\sum_{\sigma}\cos(2\pi\beta i+\phi)n_{i,\sigma}\nonumber\\
    &+U\sum_{i=1}^Nn_{i,\uparrow}n_{i,\downarrow},
\end{align}
where $\lambda$ is the strength of the on-site modulation with wavevector $2\pi\beta$ and phase offset $\phi$.
We study the model for a large, but finite, system size $L$ with periodic boundary conditions.
In the finite size setting, we note that the self-duality of the non-interacting model obtains only when the modulation frequency is $\beta=m/L$ with $m$ coprime to $L$. 
If this construction is extended to the thermodynamic limit, then the self-duality holds when $\beta$ is irrational~\cite{iyer2013mblquasiperiodic}.

In what follows, we will establish four facts about this model. First, we show that the finite commensurate approximants of the Aubry-Andr\'e-Hubbard model are dual to the Aubry-Andr\'e model with HK interactions.
Second, we show that the Aubry-Andr\'e potential can be incorporated into our generalized scheme, and is numerically tractable in the \textit{commensurate} case. 
Third, we show that for values of the Aubry-Andr\'e potential $\lambda>U/2$ a finite clustering scheme recovers the ground state energy of finite DMRG simulations of the full commensurate AAH model to $<1\%$ accuracy. 
This model therefore provides a regime in which the approximation procedure described here converges to the low-temperature physics of the thermodynamic limit even outside the unattainable limit when the interaction cluster size approaches the system size.
Finally, we show numerically that, surprisingly, non-maximal separation schemes can be competitive with maximal separation schemes even when the maximal scheme retains more momentum modes.

\subsection{AAH is dual to AAHK}
\label{sec:aah_dual_aahk}

The essential feature of the \textit{non}-interacting Aubry-Andr\'e model is that, when $\beta$ is incommensurate with the underlying lattice, the model is self-dual~\cite{aubry1980analyticity}. The same procedure can be used to prove that in the \textit{interacting} case the AAH model is dual to the AAHK model, as we will now show.

We start from the finite commensurate AAHK model with $\beta=m/L$ and $\gcd(m,L)=1$, which in position space is defined as
\begin{align}
    H_{AAHK}
    &=
    t\sum_{j,\sigma}\left(c^\dagger_{j+1,\sigma}c_{j,\sigma}+h.c.\right)
    +\lambda\sum_{j,\sigma}\cos\!\left(\frac{2\pi mj}{L}\right)n_{j,\sigma}
    \nonumber\\
    &\quad
    +\frac{U}{L}
    \sum_{R_1,\ldots,R_4}
    \delta^{(L)}_{R_1+R_3,\;R_2+R_4}\,
    c^\dagger_{R_1,\uparrow}c_{R_2,\uparrow}
    c^\dagger_{R_3,\downarrow}c_{R_4,\downarrow}.
\end{align}
To proceed, we will apply a twisted (by the modulation $m$) Fourier transform
\begin{equation}
    c_{j,\sigma}
    =
    \frac{1}{\sqrt{L}}
    \sum_{\ell=0}^{L-1}
    e^{i 2\pi m j \ell/L}c_{\ell,\sigma},
\end{equation}
which is unitary because $\gcd(m,L)=1$. Here $\ell\in\mathbb Z_L$ labels sites of the dual lattice.
This converts the hopping term into an onsite term and vice versa:
\begin{align}
    t\sum_{j,\sigma}\left(c^\dagger_{j+1,\sigma}c_{j,\sigma}+h.c.\right)
    &=
    2t\sum_{\ell,\sigma}\cos\!\left(\frac{2\pi m\ell}{L}\right)
    c^\dagger_{\ell,\sigma}c_{\ell,\sigma},\label{eq:aahk_hopping_transform}\\
    \lambda\sum_{j,\sigma}\cos\!\left(\frac{2\pi mj}{L}\right)n_{j,\sigma}
    &=
    \frac{\lambda}{2}
    \sum_{\ell,\sigma}
    \left(c^\dagger_{\ell+1,\sigma}c_{\ell,\sigma}+h.c.\right)\label{eq:aahk_onsite_transform}.
\end{align}
These are precisely the self-duality relations for the finite-size, non-interacting Aubry-Andr\'e model~\cite{aubry1980analyticity}. Applying the same transform to the HK interaction gives
\begin{align}
    &\frac{U}{L}
    \sum_{R_1,\ldots,R_4}
    \delta^{(L)}_{R_1+R_3,\;R_2+R_4}
    c^\dagger_{R_1,\uparrow}c_{R_2,\uparrow}
    c^\dagger_{R_3,\downarrow}c_{R_4,\downarrow}
    \nonumber\\
    &=
    \frac{U}{L^3}
    \sum_{R_2,R_3,R_4}
    \sum_{\ell_1,\ldots,\ell_4}
    e^{\frac{i2\pi m}{L}
    \left[
    (-\ell_1+\ell_2)R_2
    +(\ell_1-\ell_3)R_3
    +(-\ell_1+\ell_4)R_4
    \right]}
    \nonumber\\
    &\qquad\times
    c^\dagger_{\ell_1,\uparrow}c_{\ell_2,\uparrow}
    c^\dagger_{\ell_3,\downarrow}c_{\ell_4,\downarrow}
    \nonumber\\
    &=
    U\sum_{\ell_1,\ldots,\ell_4}
    \delta^{(L)}_{\ell_1,\ell_2}
    \delta^{(L)}_{\ell_1,\ell_3}
    \delta^{(L)}_{\ell_1,\ell_4}
    c^\dagger_{\ell_1,\uparrow}c_{\ell_2,\uparrow}
    c^\dagger_{\ell_3,\downarrow}c_{\ell_4,\downarrow}
    \nonumber\\
    &=
    U\sum_{\ell}
    n_{\ell,\uparrow}n_{\ell,\downarrow},\label{eq:aahk_interaction_transform}
\end{align}
where the geometric sums impose equality of the dual-lattice indices modulo $L$, since multiplication by $m$ is a permutation of $\mathbb Z_L$ when $\gcd(m,L)=1$.

Combining Eqs.~\eqref{eq:aahk_hopping_transform}--\eqref{eq:aahk_interaction_transform}, we find that the AAHK Hamiltonian can be written in the dual basis as
\begin{align}
    H_{AAHK}
    &=
    \sum_{\ell,\sigma}
    2t\cos\!\left(\frac{2\pi m\ell}{L}\right)
    c^\dagger_{\ell,\sigma}c_{\ell,\sigma}
    \nonumber\\
    &\quad
    +\frac{\lambda}{2}
    \sum_{\ell,\sigma}
    \left(c^\dagger_{\ell+1,\sigma}c_{\ell,\sigma}+h.c.\right)
    +U\sum_{\ell}n_{\ell,\uparrow}n_{\ell,\downarrow},
\end{align}
which is exactly the Aubry-Andr\'e Hubbard Hamiltonian introduced in Eq.~\eqref{eq:real_space_aah}, expressed in the dual-lattice basis, with $t$ and $\lambda/2$ exchanged.
The AAHK model is therefore unitarily equivalent to the AAH model for the finite commensurate approximants $\beta=m/L$ with $\gcd(m,L)=1$.
Thus solving the AAHK model in the localized regime is equivalent to solving the dual AAH model in the delocalized regime, and vice versa.

\subsection{Incorporating on-site modulation into the clustering scheme}
\label{sec:aah_term_in_scheme}
We will now investigate how our clustering scheme approximates the physics of the general commensurate AAH model. We consider an on-site modulation term $\hat{V}$ given in momentum space by:
\begin{equation}
    \hat{V}=\lambda/2\sum_{k}c^\dagger_{k}c_{k+Q}+h.c. , \qquad Q=2\pi\beta.
\end{equation}
For a finite commensurate chain, we parameterize the interaction separation by an integer step \(s\) on the reciprocal lattice, \(\Delta=\frac{2\pi}{L}s\), and the Aubry-Andr\'e momentum transfer by an integer step \(r\), \(Q=\frac{2\pi}{L}r\). We set $\phi=0$ without loss of generality, and omit spin indices for clarity.

We note here that, at finite $L$, the duality is established for the commensurate approximants $\beta=m/L$ with $\gcd(m,L)=1$, for which the twisted Fourier transform is unitary.
The usual irrational self-duality is then recovered in the thermodynamic limit.
In this finite commensurate setting, the twisted Fourier transform is a discrete unitary transform of the same type as \cref{eq:alpha_transform}.

\subsubsection{Consistent clustering conditions for the AAH model}

Because both the on-site modulation and the interaction couple different momenta in the AAH model, care must be taken to efficiently approximate the interaction in the clustering scheme. The on-site potential $\hat{V}$ can couple distinct interaction clusters into \emph{superclusters}. To see this, note that
repeated applications of the coupling implied by $\Delta$ and $\hat{V}$ generate the orbit $\{\ell+a s+b r\ \mathrm{mod}\ L\}$. 
This orbit has size $L/\gcd(L,s,r)$.
Therefore each supercluster contains
\begin{equation}
    \label{eq:supercluster_size}
    N_{SC}=\frac{L}{\gcd(L,s,r)}
\end{equation}
distinct momentum points. The interaction clusters alone have size $L/\gcd(L,s)$, so a clustering with $N_c$ momenta per
interaction cluster requires $N_c\mid L/\gcd(L,s)$. We note that in the case when $\gcd(L,s,r)=1$ there is only one supercluster which spans the entire system.

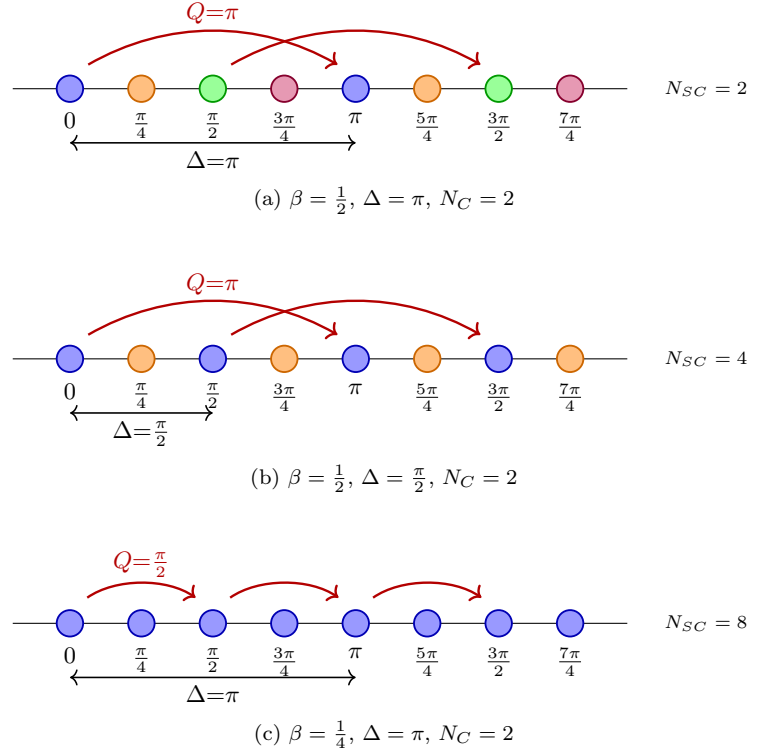
\begin{figure}[t]
\centering

\subfloat[$\beta=\frac{1}{2}$, $\Delta=\pi$, $N_C=2$\label{fig:sc_a}]{%
\begin{tikzpicture}[x=0.95cm,y=0.6cm,font=\footnotesize]
  \draw (-0.8,0) -- (7.8,0);

  \node[below,font=\small] at (0,-0.35) {$0$};
  \node[below,font=\small] at (1,-0.35) {$\frac{\pi}{4}$};
  \node[below,font=\small] at (2,-0.35) {$\frac{\pi}{2}$};
  \node[below,font=\small] at (3,-0.35) {$\frac{3\pi}{4}$};
  \node[below,font=\small] at (4,-0.35) {$\pi$};
  \node[below,font=\small] at (5,-0.35) {$\frac{5\pi}{4}$};
  \node[below,font=\small] at (6,-0.35) {$\frac{3\pi}{2}$};
  \node[below,font=\small] at (7,-0.35) {$\frac{7\pi}{4}$};

  \filldraw[fill=blue!40,draw=blue!70!black,line width=0.7pt] (0,0) circle[radius=5pt];
  \filldraw[fill=blue!40,draw=blue!70!black,line width=0.7pt] (4,0) circle[radius=5pt];
  \filldraw[fill=orange!50,draw=orange!80!black,line width=0.7pt] (1,0) circle[radius=5pt];
  \filldraw[fill=orange!50,draw=orange!80!black,line width=0.7pt] (5,0) circle[radius=5pt];
  \filldraw[fill=green!40,draw=green!60!black,line width=0.7pt] (2,0) circle[radius=5pt];
  \filldraw[fill=green!40,draw=green!60!black,line width=0.7pt] (6,0) circle[radius=5pt];
  \filldraw[fill=purple!40,draw=purple!70!black,line width=0.7pt] (3,0) circle[radius=5pt];
  \filldraw[fill=purple!40,draw=purple!70!black,line width=0.7pt] (7,0) circle[radius=5pt];

  \draw[->,red!70!black,line width=0.9pt,bend left=30,shorten <=8pt,shorten >=8pt]
    (0,0.3) to node[above,font=\small,red!70!black] {$Q{=}\pi$} (4,0.3);
  \draw[->,red!70!black,line width=0.9pt,bend left=30,shorten <=8pt,shorten >=8pt]
    (2,0.3) to (6,0.3);

  \draw[<->,black,line width=0.6pt] (0,-1.2) -- (4,-1.2);
  \node[below,font=\small] at (2,-1.2) {$\Delta{=}\pi$};

  \node[anchor=west,font=\scriptsize] at (8.2,0) {$N_{SC}=2$};
\end{tikzpicture}%
}

\vspace{0.3cm}

\subfloat[$\beta=\frac{1}{2}$, $\Delta=\frac{\pi}{2}$, $N_C=2$\label{fig:sc_b}]{%
\begin{tikzpicture}[x=0.95cm,y=0.6cm,font=\footnotesize]
  \draw (-0.8,0) -- (7.8,0);

  \node[below,font=\small] at (0,-0.35) {$0$};
  \node[below,font=\small] at (1,-0.35) {$\frac{\pi}{4}$};
  \node[below,font=\small] at (2,-0.35) {$\frac{\pi}{2}$};
  \node[below,font=\small] at (3,-0.35) {$\frac{3\pi}{4}$};
  \node[below,font=\small] at (4,-0.35) {$\pi$};
  \node[below,font=\small] at (5,-0.35) {$\frac{5\pi}{4}$};
  \node[below,font=\small] at (6,-0.35) {$\frac{3\pi}{2}$};
  \node[below,font=\small] at (7,-0.35) {$\frac{7\pi}{4}$};

  \filldraw[fill=blue!40,draw=blue!70!black,line width=0.7pt] (0,0) circle[radius=5pt];
  \filldraw[fill=blue!40,draw=blue!70!black,line width=0.7pt] (2,0) circle[radius=5pt];
  \filldraw[fill=blue!40,draw=blue!70!black,line width=0.7pt] (4,0) circle[radius=5pt];
  \filldraw[fill=blue!40,draw=blue!70!black,line width=0.7pt] (6,0) circle[radius=5pt];
  \filldraw[fill=orange!50,draw=orange!80!black,line width=0.7pt] (1,0) circle[radius=5pt];
  \filldraw[fill=orange!50,draw=orange!80!black,line width=0.7pt] (3,0) circle[radius=5pt];
  \filldraw[fill=orange!50,draw=orange!80!black,line width=0.7pt] (5,0) circle[radius=5pt];
  \filldraw[fill=orange!50,draw=orange!80!black,line width=0.7pt] (7,0) circle[radius=5pt];

  \draw[->,red!70!black,line width=0.9pt,bend left=30,shorten <=8pt,shorten >=8pt]
    (0,0.3) to node[above,font=\small,red!70!black] {$Q{=}\pi$} (4,0.3);
  \draw[->,red!70!black,line width=0.9pt,bend left=30,shorten <=8pt,shorten >=8pt]
    (2,0.3) to (6,0.3);

  \draw[<->,black,line width=0.6pt] (0,-1.2) -- (2,-1.2);
  \node[below,font=\small] at (1,-1.2) {$\Delta{=}\frac{\pi}{2}$};

  \node[anchor=west,font=\scriptsize] at (8.2,0) {$N_{SC}=4$};
\end{tikzpicture}%
}

\vspace{0.3cm}

\subfloat[$\beta=\frac{1}{4}$, $\Delta=\pi$, $N_C=2$\label{fig:sc_c}]{%
\begin{tikzpicture}[x=0.95cm,y=0.6cm,font=\footnotesize]
  \draw (-0.8,0) -- (7.8,0);

  \node[below,font=\small] at (0,-0.35) {$0$};
  \node[below,font=\small] at (1,-0.35) {$\frac{\pi}{4}$};
  \node[below,font=\small] at (2,-0.35) {$\frac{\pi}{2}$};
  \node[below,font=\small] at (3,-0.35) {$\frac{3\pi}{4}$};
  \node[below,font=\small] at (4,-0.35) {$\pi$};
  \node[below,font=\small] at (5,-0.35) {$\frac{5\pi}{4}$};
  \node[below,font=\small] at (6,-0.35) {$\frac{3\pi}{2}$};
  \node[below,font=\small] at (7,-0.35) {$\frac{7\pi}{4}$};

  \filldraw[fill=blue!40,draw=blue!70!black,line width=0.7pt] (0,0) circle[radius=5pt];
  \filldraw[fill=blue!40,draw=blue!70!black,line width=0.7pt] (1,0) circle[radius=5pt];
  \filldraw[fill=blue!40,draw=blue!70!black,line width=0.7pt] (2,0) circle[radius=5pt];
  \filldraw[fill=blue!40,draw=blue!70!black,line width=0.7pt] (3,0) circle[radius=5pt];
  \filldraw[fill=blue!40,draw=blue!70!black,line width=0.7pt] (4,0) circle[radius=5pt];
  \filldraw[fill=blue!40,draw=blue!70!black,line width=0.7pt] (5,0) circle[radius=5pt];
  \filldraw[fill=blue!40,draw=blue!70!black,line width=0.7pt] (6,0) circle[radius=5pt];
  \filldraw[fill=blue!40,draw=blue!70!black,line width=0.7pt] (7,0) circle[radius=5pt];

  \draw[->,red!70!black,line width=0.9pt,bend left=35,shorten <=8pt,shorten >=8pt]
    (0,0.3) to node[above,font=\small,red!70!black] {$Q{=}\frac{\pi}{2}$} (2,0.3);
  \draw[->,red!70!black,line width=0.9pt,bend left=35,shorten <=8pt,shorten >=8pt]
    (2,0.3) to (4,0.3);
  \draw[->,red!70!black,line width=0.9pt,bend left=35,shorten <=8pt,shorten >=8pt]
    (4,0.3) to (6,0.3);

  \draw[<->,black,line width=0.6pt] (0,-1.2) -- (4,-1.2);
  \node[below,font=\small] at (2,-1.2) {$\Delta{=}\pi$};

  \node[anchor=west,font=\scriptsize] at (8.2,0) {$N_{SC}=8$};
\end{tikzpicture}%
}%
\caption{Illustration of superclustering on an eight-point momentum grid for several choices of the Aubry-Andr\'e wavevector $Q=2\pi\beta$ and interaction-cluster spacing $\Delta$. Colored dots indicate momenta belonging to the same supercluster, red arrows denote the momentum transfer induced by $\hat V$, and the black double-headed arrow marks the interaction-cluster spacing $\Delta$. (a) For $\beta=\frac{1}{2}$ and $\Delta=\pi$ with $N_C=2$, the hopping $Q=\pi$ stays within each interaction cluster, so the supercluster size remains $N_{SC}=2$. (b) For $\beta=\frac{1}{2}$ and $\Delta=\frac{\pi}{2}$ with $N_C=2$, the hopping $Q=\pi$ connects different interaction clusters, producing two superclusters of size $N_{SC}=4$. (c) For $\beta=\frac{1}{4}$ and $\Delta=\pi$ with $N_C=2$, the hopping $Q=\frac{\pi}{2}$ links all momentum points into the $N_{SC}=8$ supercluster structure shown.}
\label{fig:clustering_with_v}
\end{figure}

Since the relative strength of the onsite potential and the hopping potential $\lambda/t$ controls the localization of electrons in the non-interacting model, we expect this parameter to likewise control the accuracy of a clustering scheme which incorporates the modulation scale. We demonstrate this fact in the following subsections.

\subsubsection{General form of the clustered AAH model}
We now derive the general form of the approximate cluster Hamiltonian in the presence of the onsite modulation. 
Using the result already derived in \cref{eq:full_cluster_ham}, we only need to derive the form of on-site modulation term $\hat{V}=\sum_{i=1}^{N}\lambda\cos(2\pi\beta i +\phi)n_i$ under the general clustering procedure outlined in \cref{sec:general_construction}. 
Re-written in the $\Delta$ clustering scheme, the modulation term becomes
\begin{align}
    \hat{V}&=\lambda/2\sum_{K}\sum_{j=1}^{N_c}c^{\dagger}_{K+j\Delta}c_{K+j\Delta+Q}+h.c.\nonumber \\
    &=\lambda/2\sum_{K}\sum_{j=1}^{N_c}c^{K,\dagger}_{j\Delta}c^{K'(K,j,Q)}_{j'(K,j,Q)\Delta}+h.c.,
\end{align}
Where here we have denoted by $K'(K,j,Q)$ and $j'(K,j,Q)$ the fact that $\hat{V}$ may, in general, hop between different interaction clusters indexed by $K$ and $K'$ and different sites within those clusters $j$ and $j'$ depending on the reference momentum $K$, the relative momentum index $j$ and the separation $\Delta$. We now apply the earlier transform \cref{eq:alpha_transform}. Notice that this transforms second-quantized operators in each cluster $K$ separately, and so we have:
\begin{align}
\label{eq:V_term_alpha_general} 
        \hat{V}
        =
        \frac{\lambda}{2N_c}
        \sum_{K}
\sum_{j,\alpha,\beta=1}^{N_c}
        &e^{-i k_j^0\cdot R^0_{\alpha}}
        e^{+i k_{j'(K,j,Q)}^0\cdot R^0_{\beta}}
        c^{K,\dagger}_{\alpha}c^{K'(K,j,Q)}_{\beta}\nonumber \\
        &+h.c. .
\end{align}
Eq.~\eqref{eq:V_term_alpha_general} is complicated by the fact that the momentum transfer $Q$ can hop between different interaction clusters $K$ and $K'$, which may in turn change the respective within-site cluster indices $j$ and $j'$ connected by the hopping. 

In the special case where the momentum hopping is always within the same cluster this expression assumes a particularly simple form. 
This is the case, for example, for any maximal separation scheme when the momentum hopping $Q$ is equal to the cluster separation, $Q=\Delta$.
In this case, $K'(K,j,Q)=K$ and $k^0_{j'(K,j,Q)}-k^0_j=Q$, a constant independent of the site indices. This allows us to simply perform the sum over within-interaction-cluster indices $j$:
\begin{align}
    \hat{V}
    &=
    \frac{\lambda}{2N_c}
    \sum_{K}\sum_{j=1}^{N_c}\sum_{\alpha,\beta=1}^{N_c}
    e^{-i k_j^0\cdot (R^0_{\alpha}-R^0_{\beta})}
    e^{iQ\cdot R^0_{\beta}}
    c^{K,\dagger}_{\alpha}c^{K}_{\beta}+h.c.\nonumber\\
    &=
    \frac{\lambda}{2}
    \sum_{K}\sum_{\alpha,\beta=1}^{N_c}
    \delta_{\alpha\beta}
    e^{iQ\cdot R^0_\beta}
    c^{K,\dagger}_{\alpha}c^{K}_{\beta}+h.c.
    \nonumber\\
    \label{eq:v_term_within_one_cluster}
    &=
    \lambda\sum_{K}\sum_{\alpha}\cos(Q\cdot R^0_{\alpha})
    n^K_{\alpha}
\end{align}
which is just the form of the original commensurate Aubry-Andr\'e interaction, but in the interaction cluster basis. This reflects the general fact that we saw already for the interacting term: if the structure factor is independent of momentum, and the interaction cluster is large enough to accommodate all coupled momenta, the transformed term takes the same form as the original term, but in the new basis.  

\begin{figure*}[tbp]
    \centering
    \includegraphics[width=0.98\linewidth]{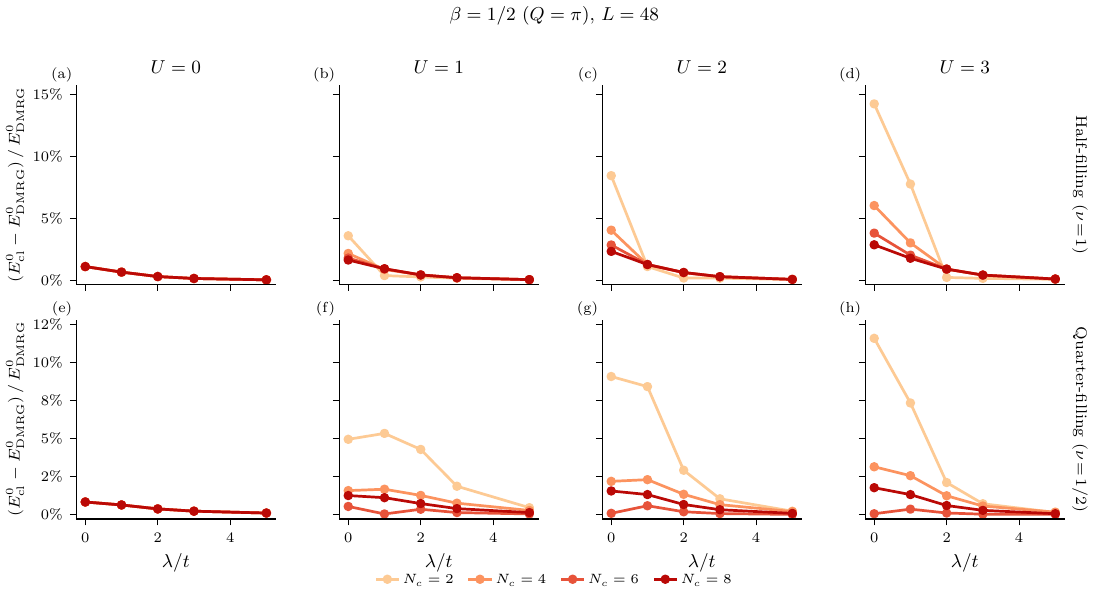}
    \caption{Effect of the Aubry-Andr\'e potential on the convergence
    of maximal-separation clustering schemes at $\beta=1/2$ ($Q=\pi$).
    Each panel shows the relative error in the ground state energy
    versus potential strength $\lambda$, with the top row at half-filling
    ($n=1$) and the bottom row at quarter-filling ($n=0.5$). Columns
    correspond to increasing interaction strength ($U=0,1,2,3$).
    Different curves show maximal-separation schemes with cluster sizes
    $N_c=2$, $4$, $6$, and $8$, all
    benchmarked against periodic finite DMRG ($\chi=128$) on an $L=48$ chain. At half-filling, the
    smallest cluster ($N_c=2$) converges rapidly with increasing $\lambda$
    and matches the $N_c=8$ result once $\lambda\gtrsim U/2$; this reflects
    the Peierls instability at $Q=2k_F=\pi$, which makes the
    $\pi$-separation cluster the dominant momentum channel. At
    quarter-filling the convergence with $\lambda$ is also present but
    requires larger $\lambda$ relative to $U$.}
    \label{fig:scheme_comparison_finite_v_t}
\end{figure*}

\subsection{Exact results for $t=0$}
 At $t=0$, or alternatively in the limit that $\lambda,U\rightarrow\infty$, we can show that the interaction cluster approximation is exact. In particular, for commensurate modulation $\beta=\tfrac{m}{n}$ with $m,n\in\mathbb{Z}$, $\gcd(m,n)=1$, and $n\mid L$, the maximal interaction clustering scheme with nearest-neighbor separation $\Delta=2\pi/n$ (and hence $N_c=n$) has a Hamiltonian unitarily equivalent to the exact Hamiltonian. 
This follows immediately from combining the result proved in \cref{sec:philip_scheme} that the maximal separation scheme maps back to the original real-space finite-site Hubbard model with the fact that in this case all hoppings are within the same cluster.
That is, for $\beta=\tfrac{m}{n}$ and $\Delta=2\pi/n$, the momentum transfer in $\hat{V}$ is given by
\(
Q=2\pi\beta=\tfrac{2\pi m}{n}=m\Delta
\),
so
\begin{equation} 
    K+Q=K+m\Delta=K+\tfrac{2\pi m}{n}=K+\tfrac{2\pi j'}{n},
\end{equation}
where the last equality follows from the fact that the clustering scheme is maximal and so any integer multiple of the interaction cluster separation $\Delta=2\pi/n$ maps back to the same cluster.
This means that all momentum space hoppings are within-cluster, and so combining the earlier results \cref{eq:v_term_within_one_cluster} with the cluster Hamiltonian \cref{eq:full_cluster_ham} gives, at $t=0$,
\begin{equation}
    H(t=0)=
    \sum_{K\in\mathcal K}\sum_{\alpha=1}^{n}
    \left[
    \lambda \cos(Q\cdot R^0_{\alpha})n^K_{\alpha}
    +Un^K_{\alpha,\uparrow}n^K_{\alpha,\downarrow}
    \right].
\end{equation}
This is just equal to \(L/n\) copies of the Hamiltonian within an interacting cluster, which is exactly equal to the finite commensurate \(t=0\) AAH model written in the \(\alpha\) basis. 
Explicitly, associate each of the $L/N_c$ cluster representative $K$ to a coarse lattice index $X$, and each $\alpha$ to a within-cluster site index. We may then relabel
\begin{equation}
c^{K_X}_{\alpha,\sigma}\to c_{XN_c+\alpha,\sigma}, \qquad
n^{K_X}_{\alpha,\sigma}\to n_{XN_c+\alpha,\sigma}.
\end{equation}
Moreover, in one dimension we have
\begin{equation}
    \cos(Q\cdot R^0_\alpha)
    =\cos(2\pi\beta\alpha)
    =\cos(2\pi\beta(XN_c+\alpha)),
\end{equation}
where the last equality uses the maximal-separation condition $N_c=n$ together with \(\beta=m/n\), so that
\(
2\pi\beta XN_c
=2\pi\tfrac{m}{n}Xn
=2\pi Xm\in 2\pi\mathbb{Z}
\)
and hence does not change the value of the cosine. Under this relabeling, the $t=0$ AAH Hamiltonian becomes
\begin{align}
    &H(t=0)
    =
    \sum_{X=0}^{L/n-1}\sum_{\alpha=1}^{n}
    \bigg[
    \lambda \cos\!\left(2\pi\beta(XN_c+\alpha)\right)n_{Xn+\alpha} \nonumber \\
    &+ U n_{Xn+\alpha,\uparrow} n_{Xn+\alpha,\downarrow}
    \bigg].
\end{align}
Relabeling into the real-space lattice coordinate \(i=Xn+\alpha\) then gives the original real-space Hamiltonian:
\begin{equation}
    H(t=0)=\lambda\sum_{i=1}^{L}\cos(2\pi\beta i)n_{i}+U\sum_{i=1}^{L}n_{i,\uparrow}n_{i,\downarrow}.
\end{equation}

\begin{figure*}[tbp]
    \centering
    \includegraphics[width=0.98\linewidth]{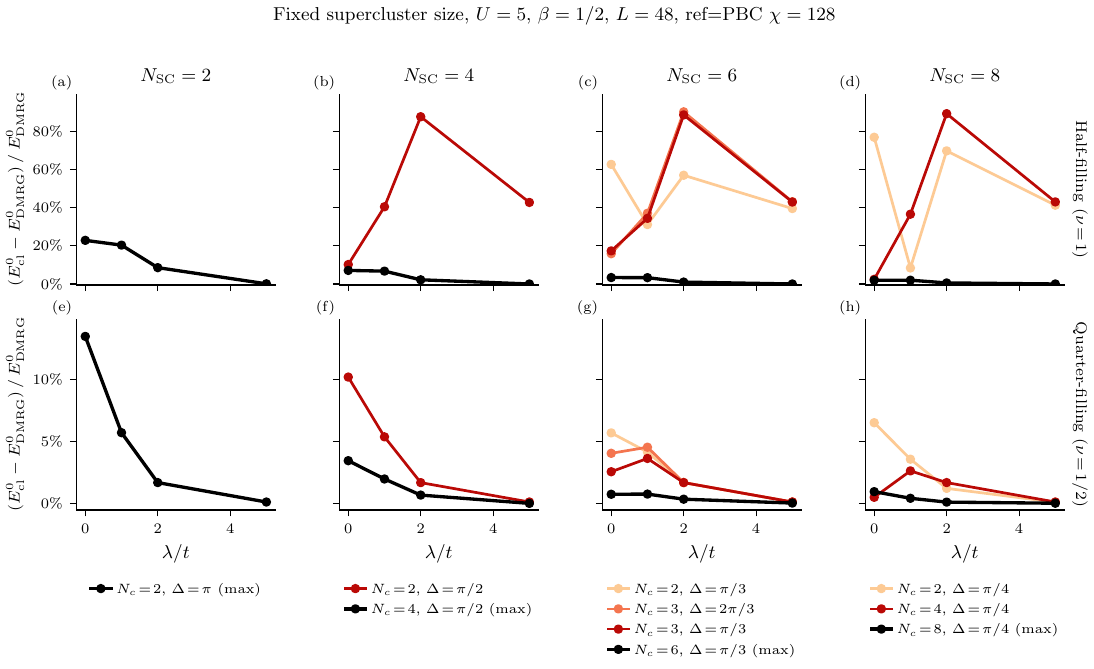}
    \caption{Comparison of maximal versus non-maximal clustering
    schemes at fixed supercluster size $N_{SC}$ and fixed interaction
    strength $U=5$, with $\beta=1/2$. Each column fixes the
    supercluster size (from $N_{SC}=2$ to $N_{SC}=8$); the top row
    shows half-filling ($n=1$) and the bottom row quarter-filling
    ($n=0.5$). Within each panel, different curves correspond to
    different combinations of interaction cluster size $N_c$ and
    spacing $\Delta$ that yield the same $N_{SC}$---and hence the same
    computational cost---with the maximal-separation scheme shown in
    black. At $N_{SC}=4$ and $N_{SC}=6$, a non-maximal $N_c=2$ scheme
    (whose interaction clusters are fused into a larger supercluster by
    the onsite potential) becomes competitive with the corresponding maximal scheme
    at large $\lambda$ values, despite retaining fewer interacting
    momentum modes within each interaction cluster. All results are
    benchmarked against PBC finite DMRG ($\chi=128$) on an $L=48$ chain.}
    \label{fig:v_convergence_fixedsc_v_pi}
\end{figure*}

\subsection{$\hat{V}$ accelerates convergence at finite $t$}
At finite $t$, we expect that the Aubry-Andr\'e potential $\hat{V}$ will accelerate convergence of the clustering scheme to the exact ground state energy. In this section, we verify numerically that this is the case at half- and quarter-filling. 
In general, in the presence of the onsite potential $\hat{V}$, the supercluster size determines the computational cost of exact diagonalization of the resulting Hubbard clusters.
In this section, we consider the family of maximal separation schemes defined by requiring that the $\hat{V}$ hopping always falls within the same interaction cluster, which sets the supercluster size equal to the interaction cluster size $N_{SC}=N_c$. In the next section we relax this condition and consider alternative schemes at fixed supercluster sizes.

In the presence of an onsite potential $\hat{V}$ there is no corresponding exact solution. Instead, we use finite-system DMRG with periodic couplings as implemented in \texttt{TeNPy}~\cite{hauschild2018tenpy}. Although periodic-boundary DMRG is less efficient than the corresponding open-boundary calculation and therefore does not scale as favorably with bond dimension $\chi$~\cite{verstraete2004dmrgpbc,pippan2010efficientpbc}, the chain sizes considered here are moderate, so this constraint is not severe in practice. We verified convergence of the DMRG reference with increasing $\chi$, and also confirmed that the PBC calculation gives a lower ground-state energy than the corresponding OBC run, as expected for the finite periodic system that is most directly comparable to our clustering construction.

The maximal interaction separation at a given $\beta=m/n$, with $m,n$ coprime, is set as $\Delta^{\text{max}}_{\beta}=2\pi/n$ with interaction cluster size set to $n$.
The family of allowed maximal separation schemes is then the set of separations obtained by dividing $\Delta^{\text{max}}$ by any positive integer,
\[
    \Delta=\frac{\Delta^{\text{max}}}{Z}=\frac{2\pi}{nZ},
    \qquad
    N_c=nZ,
\]
with the finite-size divisibility condition \(nZ \mid L\).

The resulting comparisons to DMRG are shown in \cref{fig:scheme_comparison_finite_v_t} for $\beta=1/2$ (\(Q=\pi\)). Additionally in Appendix~\ref{sec:app_numerical_results} we show the results for $\beta=1/3$ (\(Q=2\pi/3\)), and $\beta=1/4$ (\(Q=\pi/2\)) in \cref{fig:scheme_comparison_finite_v_t_2pi3,fig:scheme_comparison_finite_v_t_1pi4}, respectively. The same qualitative picture holds in each case.

For $N_c=2$ and $\beta=1/2$ (\(Q=\pi\)), this effect is particularly striking.
At half-filling, shown in the top row of \cref{fig:scheme_comparison_finite_v_t}, we see that the two-site clustering reaches the performance of the eight-site clustering scheme when $U<\lambda$. 
This reflects the Peierls instability to a modulation $Q=2k_F=\pi$, which couples the opposite Fermi points of the non-interacting system. The convergence is slower for $\beta=1/3$ (\cref{fig:scheme_comparison_finite_v_t_2pi3}) and $\beta=1/4$ (\cref{fig:scheme_comparison_finite_v_t_1pi4}), because these modulations do not directly couple the two Fermi points.

At quarter-filling for $\beta=1/2$ (\(Q=\pi\)), shown in the bottom row of \cref{fig:scheme_comparison_finite_v_t}, the two-site clustering scheme also converges to the larger-site clustering as we increase $\lambda$. In Appendix~\ref{sec:app_numerical_results} we also show (see \cref{fig:scheme_comparison_finite_v_t_highU}) that for sufficiently large $U$, the two-site clustering scheme can outperform the largest $N_c=8$ clustering scheme we consider here, and that this outperformance increases as we increase $\lambda$.

\subsection{Non-maximal separation schemes at finite $t$}

Having shown in the previous subsection that increasing the strength of the Aubry-Andr\'e modulation $\hat V$ can give excellent agreement with the DMRG energy for small clusters, we now ask a sharper question: at fixed computational cost, is the maximal separation scheme also the best choice? 
In the presence of the onsite potential, the relevant cost is set by the
supercluster size $N_{SC}$. We therefore compare the maximal scheme with
$N_c=N_{SC}$ to non-maximal schemes with smaller interaction clusters
$N_c<N_{SC}$ for which the Aubry-Andr\'e hopping $\hat V$ fuses the interaction clusters into superclusters of the same size $N_{SC}$.

A simple mode-counting argument would suggest that the maximal scheme should
perform best in this comparison. At fixed $N_{SC}$, the maximal scheme retains the largest set of interaction channels within the supercluster, whereas a fused scheme treats only a smaller subset of those momenta as directly interacting. If the quality of the approximation were controlled only by the number of retained interaction modes, the maximal scheme would therefore be optimal.

The data in \cref{fig:v_convergence_fixedsc_v_pi} show that this expectation is largely correct. At half-filling, the maximal scheme remains the most
accurate across the parameter range shown. At quarter-filling, the
fused $N_c=2$ schemes for $N_{SC}=4$ and $N_{SC}=6$ achieve competitive performance with the cost-matched maximal scheme over an intermediate window of $\lambda$, even though they retain fewer interacting momentum modes. This suggests that once $\hat V$ selects a momentum scale, fixed-cost accuracy is not determined solely by raw mode count. As in the $V=0$ Hubbard comparison, the organization of the retained channels matters, and understanding how to choose the optimal clustering at fixed $N_{SC}$ is an interesting direction for further work.

\section{Conclusion}

In this work, we have introduced a general clustering scheme that allows us to approximate the Hubbard model by preserving selected momentum channels in the interaction Hamiltonian. This generalizes previous work on HK models, which either retains only a single momentum channel (and therefore corresponds to our $N_c=1$ case), or preserves only maximally separated momentum interactions as in the MMHK construction.
We showed that our construction recovers the twist-averaged finite-site Hubbard model, with, in general, all-all hoppings $J_{\alpha\beta}$.
In particular, we showed that the maximal scheme is the special case when the approximate model is a finite-site Hubbard model of the same form as the original Hamiltonian. 
This clarifies that the surprising phenomenological success of HK models in reproducing a wide range of Mott physics arises from their equivalence to the finite-site Hubbard models with twist averaging.
We then show that there are alternative clustering schemes available from our general construction that numerically outperform the maximal scheme in reproducing the ground state energy of the Hubbard model, especially at quarter-filling.

In the original Hubbard model, no particular momentum mixing channels in the interacting term are favored. Our general scheme also allows us to extend to the case in which an onsite potential leads to an emergent superclustering by favoring certain momentum mixing channels.
In this case, it is most important to capture momentum couplings within the effective reciprocal supercell. We are hence able to show that increasing the strength of the onsite potential can lead to the convergence of surprisingly small Hubbard superclusters. We demonstrated this through a study of the Aubry-Andr\'e Hubbard model
 
Our work can be extended in three main directions. First, the Aubry-Andr\'e Hubbard model is effectively the one-dimensional analogue of the physics that arises in the Bistritzer-MacDonald model of twisted bilayer graphene. Extending our clustering scheme to two dimensions would allow the Hubbard interaction to be introduced into the BM model on an equal footing to non-interacting terms and is hence a natural application of the framework given here. Second, we can only numerically access small cluster sizes in the commensurate AAH model. To extend to the incommensurate case requires developing an approximate treatment of intercluster hopping terms. Finally, it would be interesting to understand the connection between the momentum mixing given by the scheme described here to the momentum mixing implicit in a Bethe Ansatz treatment of the standard Hubbard model.

\begin{acknowledgments}
The authors thank P. Phillips and L.K. Wagner for helpful discussions. This work was supported by the U.S. National Science Foundation under grant No. DMR-2510219.
\end{acknowledgments}

\clearpage   
\onecolumngrid
\appendix

\section{Real space form of the interaction for minimal separation}
\label{sec:app_minimal_sep_real_space}
In this section we consider the clustering scheme where the separation takes its minimal value, $\Delta=2\pi/L$, as opposed to the ``maximal'' clustering scheme analyzed in Sec.~\ref{sec:philip_scheme}. 
We use a ``symmetric parameterization'' in which each interaction cluster contains an odd number of
modes $N_c=2S+1$, and we place the cluster representative $K$ at the center of the cluster.
As before, we start from the cluster truncated Hamiltonian \cref{eq:cluster_truncated_interaction}:
\begin{equation}
    H_{\mathrm{int}}
    =
    \frac{U}{N_c}\sum_{\mathbf{K}\in\mathcal{K}}
    \sum_{\substack{\mathbf{k}_{1},\mathbf{k}_{2},\mathbf{q}\in \mathcal C_{\mathbf 0}}}
    c^\dagger_{\mathbf{K}+(\mathbf{k}_{1}\oplus\mathbf{q}),\uparrow}c_{\mathbf{K}+\mathbf{k}_{1},\uparrow}
    \,c^\dagger_{\mathbf{K}+(\mathbf{k}_{2}\ominus\mathbf{q}),\downarrow}c_{\mathbf{K}+\mathbf{k}_{2},\downarrow}.
\end{equation}

Define the number of clusters $N_X:=L/N_c$ and a cluster index $M\in\{0,1,\dots,N_X-1\}$, with
\begin{equation}
    K_M := \frac{2\pi}{L}(N_c M + S).
\end{equation}
Within cluster $M$ we parameterize the within-cluster momenta as:
\begin{equation}
    k_1 = K_M + n\Delta,\qquad k_2 = K_M + m\Delta,\qquad n,m\in \mathbb{Z}_{N_c}.
\end{equation}
Throughout, we identify $\mathbb{Z}_{N_c}$ with the symmetric representatives
$S_N:=\{-S,-S+1,\dots,0,\dots,S-1,S\}$, and we use $\oplus,\ominus$ for addition/subtraction modulo $N_c$. Fourier transforming the approximate cluster Hamiltonian into real space, we write the interaction as
\begin{align}
    H_{\mathrm{int}}
    &= \sum_{R_1,\dots,R_4}\sum_{M=0}^{N_X-1}
    e^{iK_M A}\;
    \mathcal{I}(R_1,R_2;R_3,R_4)\;
    c^\dagger_{R_1,\uparrow}c_{R_2,\uparrow}c^\dagger_{R_3,\downarrow}c_{R_4,\downarrow},
\end{align}
with 
\begin{equation}
    A := R_1+R_3-R_2-R_4,
\end{equation}
and  $\mathcal{I}$ defined as the within-cluster sum
\begin{equation}
    \mathcal{I}(R_1,R_2;R_3,R_4)
    :=\sum_{n,m,q\in\mathbb{Z}_{N_c}}
    \exp\!\Big(i\Delta\big[(n\oplus q)R_1-nR_2+(m\ominus q)R_3-mR_4\big]\Big).
\end{equation}
We then perform each sum in turn, starting with the inner sum.
Using the definition of $K_M$,
\begin{align}
    \sum_{M=0}^{N_X-1}e^{iK_M A}
    &= e^{i\frac{2\pi}{L}SA}\sum_{M=0}^{N_X-1}e^{i\frac{2\pi}{N_X}MA}
     = e^{i\frac{2\pi}{L}SA}\;N_X\,\delta_{A\equiv 0\ (\mathrm{mod}\ N_X)}.
\end{align}
To deal cleanly with the wrap-around indices in $n\oplus q$ and $m\ominus q$, we use the standard
discrete convolution theorem on $\mathbb{Z}_{N_c}$.  Let $\omega:=e^{i2\pi/N_c}$ and define
\begin{equation}
    C_q(R_A,R_B)
    :=\sum_{n\in\mathbb{Z}_{N_c}} e^{i\Delta(n\oplus q)R_A}\,e^{-i\Delta n R_B}.
\end{equation}
This is a discrete convolution $C_q=\sum_n f_{n\oplus q}\,\overline{g_n}$ with
\begin{equation}
    f_n(R_A):=e^{i\Delta nR_A},\qquad g_n(R_B):=e^{i\Delta nR_B}
\end{equation}
Let $F_p(R_A),G_p(R_B)$ denote their $N_c$-point DFTs:
\begin{align}
    &F_p(R_A):=\sum_{n\in\mathbb{Z}_{N_c}} f_n(R_A)\,\omega^{-pn},\qquad \nonumber \\
    &G_p(R_B):=\sum_{n\in\mathbb{Z}_{N_c}} g_n(R_B)\,\omega^{-pn}.
\end{align}
Then the discrete convolution theorem gives
\begin{equation}
    C_q(R_A,R_B)=\frac{1}{N_c}\sum_{p=0}^{N_c-1}F_p(R_A)\,\overline{G_p(R_B)}\,\omega^{pq}.
\end{equation}
With the symmetric representative set $n=-S,\dots,S$ (so $N_c=2S+1$), these DFTs are Dirichlet kernels:
\begin{align}
    F_p(R_A)
    &=\sum_{n=-S}^{S}e^{in(\Delta R_A-2\pi p/N_c)}
      =:D_S\!\left(\Delta R_A-\frac{2\pi p}{N_c}\right),\\
    G_p(R_B)
    &=\sum_{n=-S}^{S}e^{in(\Delta R_B-2\pi p/N_c)}
      =:D_S\!\left(\Delta R_B-\frac{2\pi p}{N_c}\right),
\end{align}
where $D_S(\theta):=\sum_{n=-S}^{S}e^{in\theta}=1+2\sum_{r=1}^{S}\cos(r\theta)$. Since this is real,
the conjugate $\overline{G_p(R_B)}=G_p(R_B)$.

Noting that the $n$- and $m$-dependent parts factorize, we can write
\begin{equation}
    \mathcal{I}(R_1,R_2;R_3,R_4)=\sum_{q\in\mathbb{Z}_{N_c}} C_q(R_1,R_2)\,C_{-q}(R_3,R_4).
\end{equation}
Substituting the DFT form and using $\sum_{q\in\mathbb{Z}_{N_c}}\omega^{(p-p')q}=N_c\,\delta_{p,p'}$ yields
\begin{align}
    \mathcal{I}(R_1,R_2;R_3,R_4)
    &=\frac{1}{N_c}\sum_{p=0}^{N_c-1}
      \Bigg[\prod_{i=1}^{4}D_S\!\left(\Delta R_i-\frac{2\pi p}{N_c}\right)\Bigg],
\end{align}
where $R_i$ stands for $R_1,R_2,R_3,R_4$ respectively in the product.

Noting that $D_S(\Delta R-\frac{2\pi p}{N_c})=D_S(\frac{2\pi}{L}(R-pN_X))$, we can see that each factor is peaked when
$R\simeq pN_X$ (mod $L$). Thus, the wrapped minimal-separation interaction can be viewed as follows:
the interaction is local in the \emph{cluster-orbital} basis (i.e., \cref{eq:diagonal_alpha_interaction}), but when expressed in real space it becomes a long-ranged, oscillatory four-fermion term whose spatial envelope is set by the Dirichlet kernel width $\sim N_X$.

From the momentum-space perspective, retaining only consecutive modes corresponds to a channel selection biased toward
\emph{small} momentum transfers within each cluster; large-$q$ scattering processes (such as $2k_F$ backscattering or
Umklapp at half filling) are only recovered once $N_c$ is large enough that the relevant transfers fit inside a block.
Accordingly, at small $N_c$ this truncation preferentially captures forward-scattering processes. Taking the limit where the cluster size goes to system size $N_c\to L$, 
each factor $D_S(\Delta R_i-2\pi p/N_c)$ becomes increasingly localized - in the limit converging to a delta function. For finite values of the cluster size $N_c$ it yields a controlled, oscillatory smearing with range set by $N_c a$.

\section{Comparison with Momentum-Mixing HK model}
\label{sec:mmhk_notation_comparison}
To help with comparison to the literature we make explicit here how the maximal scheme considered here recovers the ``Momentum-Mixing'' HK model. Starting from the maximal Hamiltonian, \cref{eq:full_cluster_ham}:
\begin{align}
    H=\sum_{\mathbf K\in\mathcal K}&\left[
    \sum_{a=1}^{z}t_a e^{i\mathbf{K}\cdot\mathbf{r}_a}
    \sum_{\alpha,\beta=1}^{N_c}\sum_{\sigma}
    J^{a}_{\alpha\beta}c^{\dagger,\mathbf K}_{\alpha,\sigma}c^{\mathbf K}_{\beta,\sigma}
    +\mathrm{h.c.}\right.
    \left.+U\sum_{\alpha=1}^{N_c}n^{\mathbf K}_{\alpha,\uparrow}n^{\mathbf K}_{\alpha,\downarrow}
    \right].
\end{align}
In our notation the full term $e^{i\mathbf{K}\cdot\mathbf{r}_a}J_{\alpha\beta}+h.c.$ is equivalent to the matrix $g_{\alpha,\alpha'}$ in~\cite{Mai2025momentummixings}. Explicitly, the case considered there was the two-dimensional $N_c=4$ case with two nearest-neighbor hoppings of equal strength $t_{1}=t_2=t$, and corresponding hopping vectors $\mathbf{r}_{1}=a\hat{\mathbf{e}}_x,\mathbf{r}_2=a\hat{\mathbf{e}}_{y}$ and two next-nearest-neighbor hoppings of equal strength, $t_3=t_4=t'$ and corresponding vectors, $\mathbf{r}_3=a(\hat{\mathbf{e}}_x+\hat{\mathbf{e}}_y),\mathbf{r}_4=a(-\hat{\mathbf{e}}_x + \hat{\mathbf{e}}_y)$. In that case, the hopping matrices and phase factors become\footnote{Note that we are imposing periodic boundary conditions on the hoppings \textit{within} a cluster}:
\begin{align}
&
t e^{i\mathbf{K}_x a}
\begin{pmatrix}
0&1&0&0\\
1&0&0&0\\
0&0&0&1\\
0&0&1&0\\
\end{pmatrix},
t e^{i\mathbf{K}_y a}
\begin{pmatrix}
0&0&1&0\\
0&0&0&1\\
1&0&0&0\\
0&1&0&0\\
\end{pmatrix} \nonumber \\
& t' e^{i(K_x+K_y) a}
\begin{pmatrix}
0&0&0&1\\
0&0&1&0\\
0&1&0&0\\
1&0&0&0\\
\end{pmatrix},
t' e^{i(-K_x+K_y) a}
\begin{pmatrix}
0&0&0&1\\
0&0&1&0\\
0&1&0&0\\
1&0&0&0\\
\end{pmatrix}
\end{align}
And so the total cluster kinetic energy is:
\begin{equation}
    T_\mathbf{K}=\sum_{a=1}^{4}\sum_{\alpha\beta}e^{i\mathbf{K}\cdot\mathbf{r}_a}J^a_{\alpha\beta}+e^{iK_y\cdot a}J^a_{\alpha\beta}h.c.=
    \begin{pmatrix}
    0&\varepsilon_x&\varepsilon_y&\varepsilon_{xy}\\
    \varepsilon_x&0&\varepsilon_{xy}&\varepsilon_y\\
    \varepsilon_y&\varepsilon_{xy}&0&\varepsilon_x\\
    \varepsilon_{xy}&\varepsilon_{y}&\varepsilon_x&0\\
    \end{pmatrix},
\end{equation}
with $\varepsilon(x)=-2t\cos(K_x a),\varepsilon(y)=-2t\cos(K_y a), \varepsilon_{xy}=-4t'\cos(K_x a)\cos(K_y a)$, reproducing Eq. (7) of~\cite{Mai2025momentummixings}.

\section{Additional numerical results}
\label{sec:app_numerical_results}

In this appendix we present supplementary numerical data for the cluster approximation introduced in the main text. The main text figures show the relative error of the ground state energy relative to exact or approximately exact benchmarks; here we complement them with absolute energies, additional data near weak coupling, and filling curves $\nu(\mu_0)$ that probe the thermodynamic response. We organize the results by figure: the one-dimensional Hubbard benchmarks (\cref{sec:app_hubbard}), the AAH convergence at $\beta=1/2$ (\cref{sec:app_aah}), and the fixed supercluster comparison (\cref{sec:app_fixedsc}).

\subsection{Additional data for Hubbard benchmarks (\texorpdfstring{$\lambda=0$}{lambda=0})}
\label{sec:app_hubbard}

\Cref{fig:hub_abs_energy} shows the absolute energy per site for the pure one dimensional Hubbard model ($\lambda=0$) as a function of $U/t$, corresponding to the relative error data in \cref{fig:scheme_comparison_v0}. Even at moderate cluster sizes, the cluster energies closely track the exact result across the full range of $U$.

\begin{center}
    \includegraphics[width=0.98\linewidth]{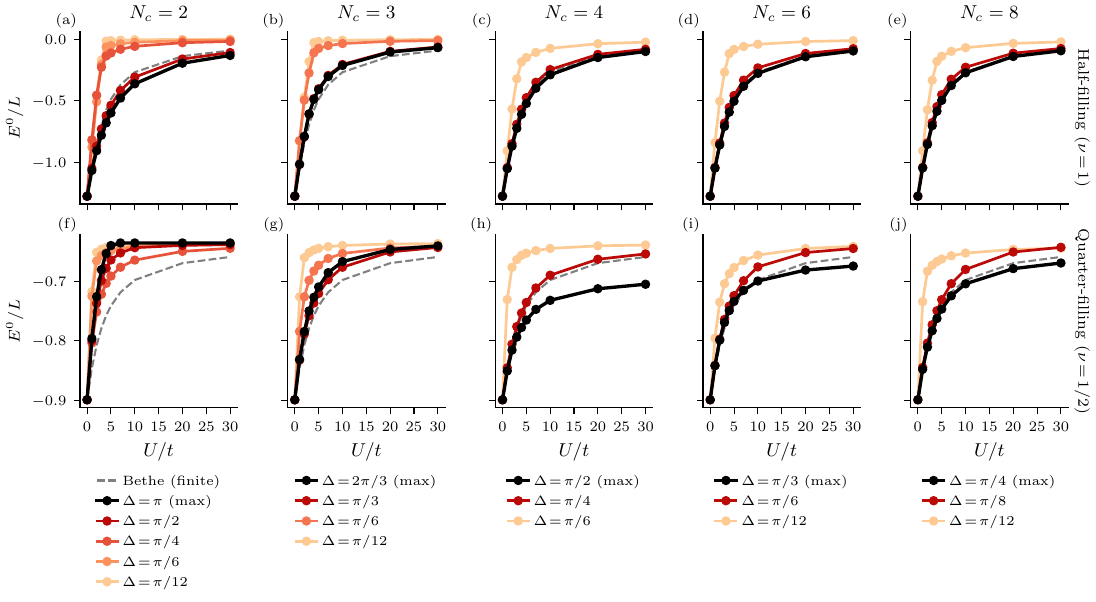}
    \captionof{figure}{Absolute ground state energy per site $E^0/L$ as a function of $U/t$ for $\lambda=0$, $L=48$. Each column corresponds to a cluster size $N_c$; top row is half-filling, bottom row is quarter-filling. Colored lines show different interaction cluster separations $\Delta$ (maximal in black); the grey dashed line is the exact finite periodic Bethe ansatz, consistent with the reference used in \cref{fig:scheme_comparison_v0}.}
    \label{fig:hub_abs_energy}
\end{center}

Ground state energy is typically the first quantity to converge in approximate methods. As a more stringent test, we compare the filling $\nu(\mu_0)$ obtained by sweeping the chemical potential at fixed $U$. \Cref{fig:compressibility_hub} shows these filling curves for the full $U$ range. At large $U$, the staircase structure of the (Mott) insulating plateaus is clearly resolved; different interaction separations $\Delta$ are compared in each $(U, N_c)$ panel.

\begin{center}
    \includegraphics[width=0.90\linewidth]{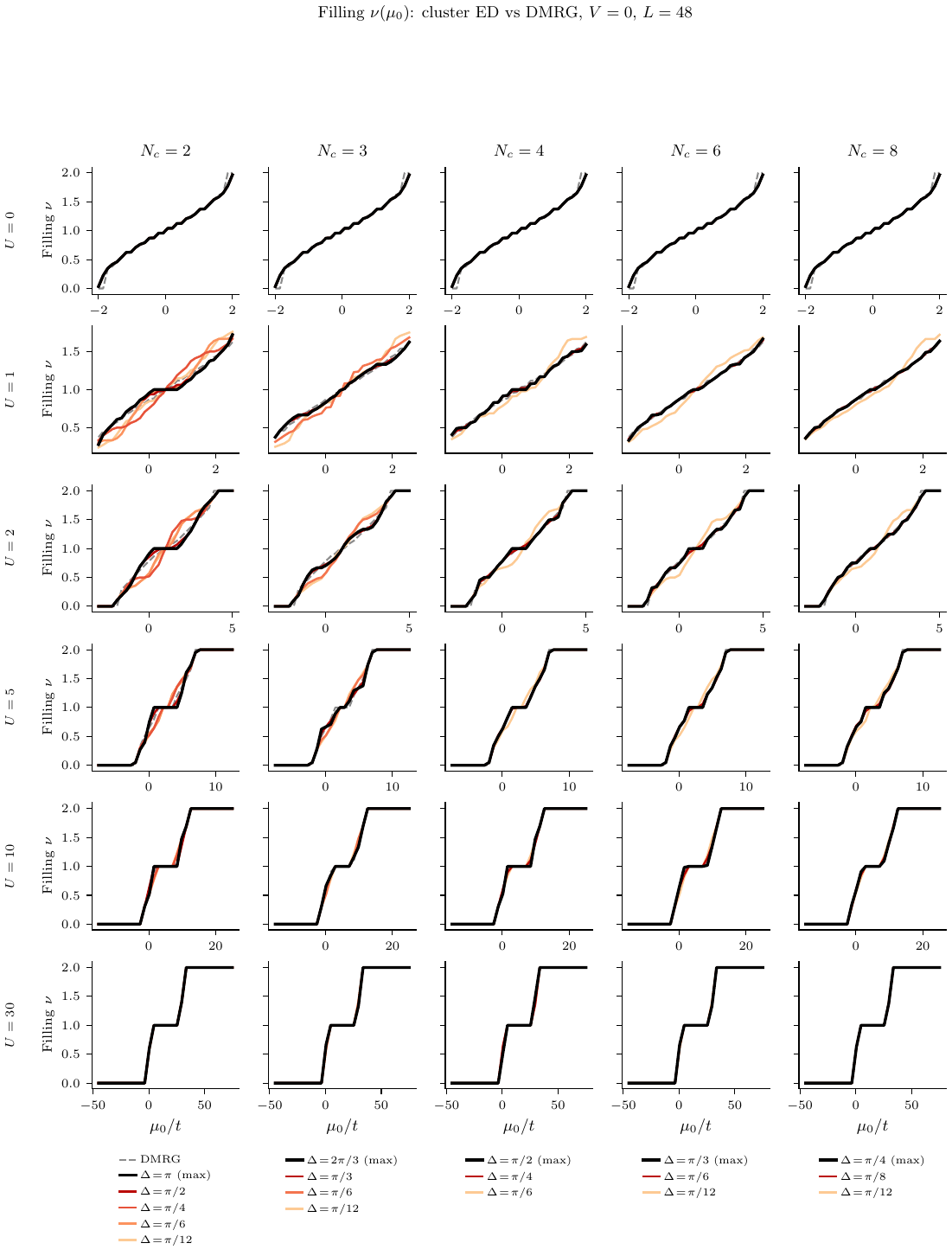}
    \captionof{figure}{Filling $\nu(\mu_0)$ for $\lambda=0$, $L=48$. Rows correspond to $U$ values; columns to cluster sizes $N_c$. Colored lines show different interaction separations $\Delta$ (maximal in black); grey dashed is OBC DMRG ($\chi=32$).}
    \label{fig:compressibility_hub}
\end{center}

\Cref{fig:hub_filling_relerr} shows the relative error of the filling. The largest deviations occur near the edges of the Mott plateau, where the filling changes rapidly and the cluster approximation must accurately resolve the charge gap.

\begin{center}
    \includegraphics[width=0.98\linewidth]{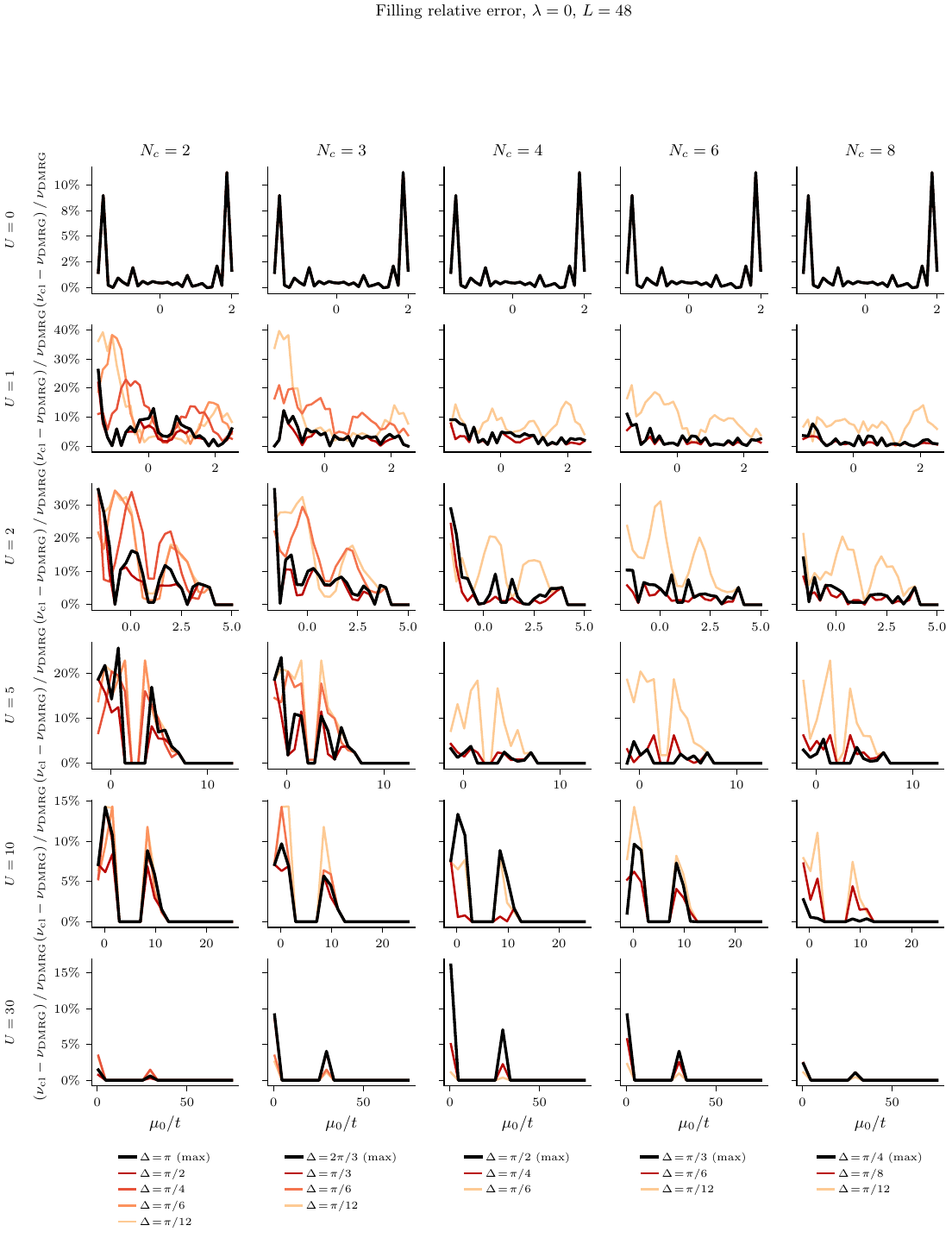}
    \captionof{figure}{Relative filling error $|\nu_{\mathrm{cl}} - \nu_{\mathrm{DMRG}}|/\nu_{\mathrm{DMRG}}$ vs $\mu_0/t$ for $\lambda=0$, $L=48$. Same layout as \cref{fig:compressibility_hub}. Reference is OBC DMRG ($\chi=32$).}
    \label{fig:hub_filling_relerr}
\end{center}

\subsection{Additional data for AAH convergence (\texorpdfstring{$\beta=1/2$}{beta=1/2})}
\label{sec:app_aah}

\Cref{fig:v_abs_energy_1-2} shows the absolute energy as a function of $U/t$ at $\beta=1/2$, with columns corresponding to different values of the Aubry-Andr\'e potential $\lambda$. This complements the main text \cref{fig:scheme_comparison_finite_v_t}, which sweeps $\lambda$ at fixed $U$, by showing the $U$-dependence at fixed $\lambda$.

\begin{center}
    \includegraphics[width=0.98\linewidth]{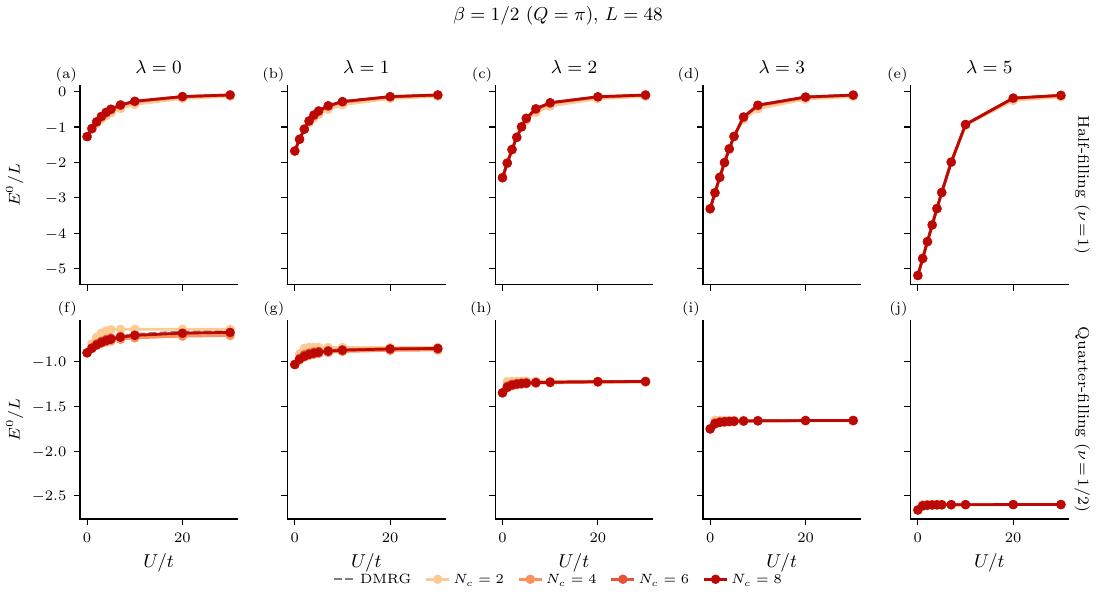}
    \captionof{figure}{Absolute ground state energy per site $E^0/L$ as a function of $U/t$ for $\beta=1/2$, $L=48$. Columns correspond to different $\lambda$ values; top row is half-filling, bottom row is quarter-filling. Lines show different cluster sizes $N_c$; grey dashed is finite OBC DMRG ($\chi=32$).}
    \label{fig:v_abs_energy_1-2}
\end{center}

\Cref{fig:scheme_comparison_finite_v_t_highU} extends the $\beta=1/2$ analysis to the large-$U$ regime ($U=5,7,10,20$), where charge fluctuations are suppressed and the system approaches the Mott insulating limit.

\begin{center}
    \includegraphics[width=0.98\linewidth]{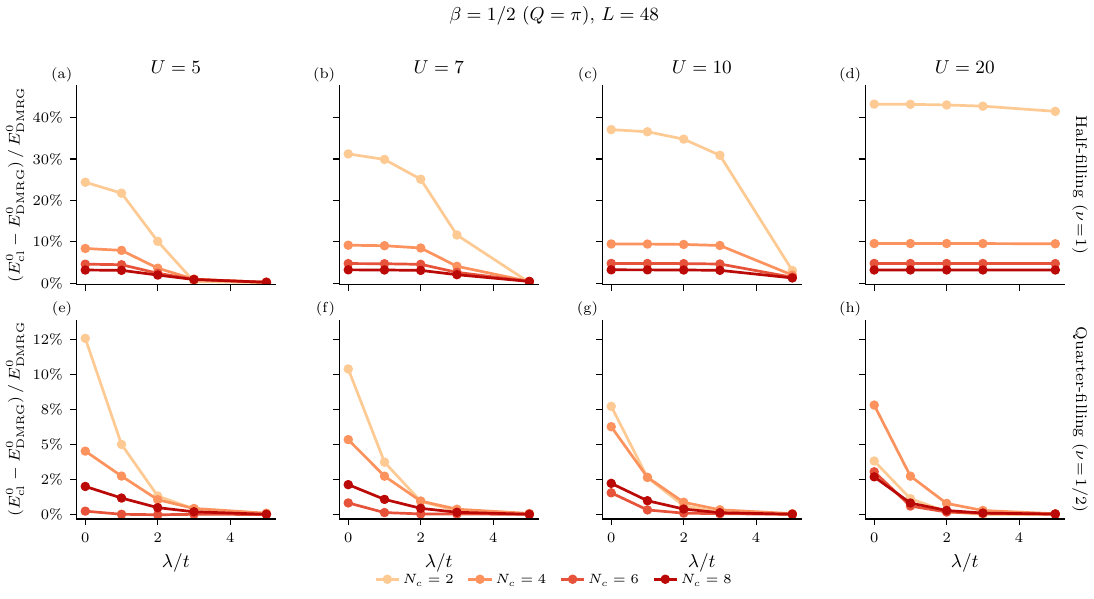}
    \captionof{figure}{Relative energy error as a function of $\lambda/t$ at $\beta=1/2$ in the large-$U$ regime, $L=48$. Top row: half-filling. Bottom row: quarter-filling. Reference is OBC DMRG ($\chi=32$).}
    \label{fig:scheme_comparison_finite_v_t_highU}
\end{center}

We also show the convergence at additional modulation ratios. \Cref{fig:scheme_comparison_finite_v_t_2pi3} and \cref{fig:v_abs_energy_1-3} show respectively the relative and absolute energies for $\beta=1/3$ ($Q=2\pi/3$); \cref{fig:scheme_comparison_finite_v_t_1pi4} and \cref{fig:v_abs_energy_1-4} show $\beta=1/4$ ($Q=\pi/2$). The convergence pattern is qualitatively similar across all $\beta$ values.

\begin{center}
    \includegraphics[width=0.98\linewidth]{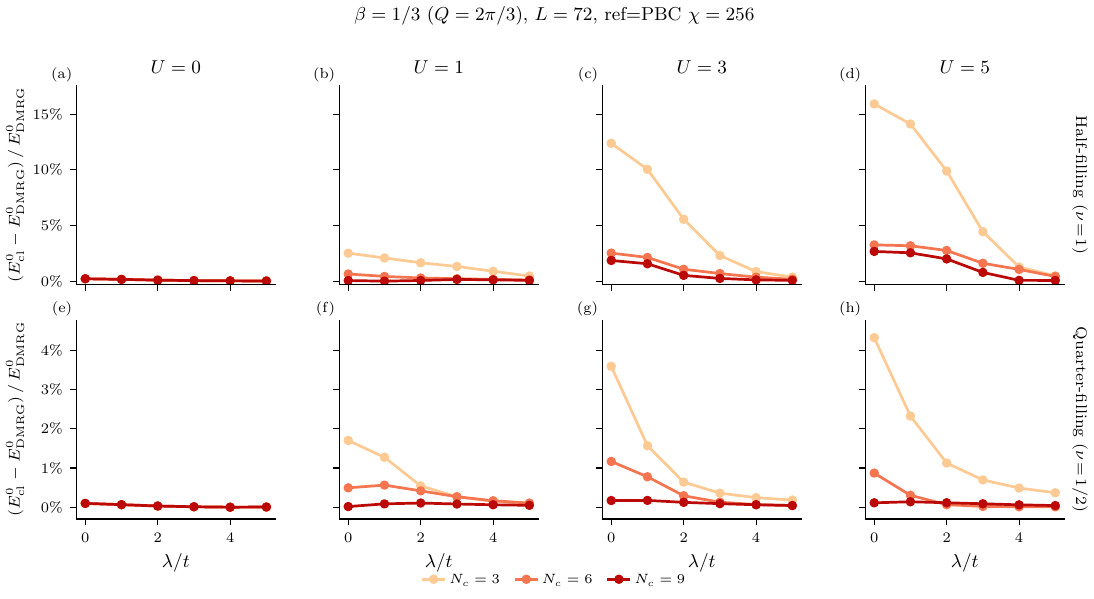}
    \captionof{figure}{Relative energy error as a function of $\lambda/t$ at $\beta=1/3$ ($Q=2\pi/3$), $L=72$. Top row: half-filling. Bottom row: quarter-filling. Reference is PBC DMRG ($\chi=256$).}
    \label{fig:scheme_comparison_finite_v_t_2pi3}
\end{center}

\begin{center}
    \includegraphics[width=0.98\linewidth]{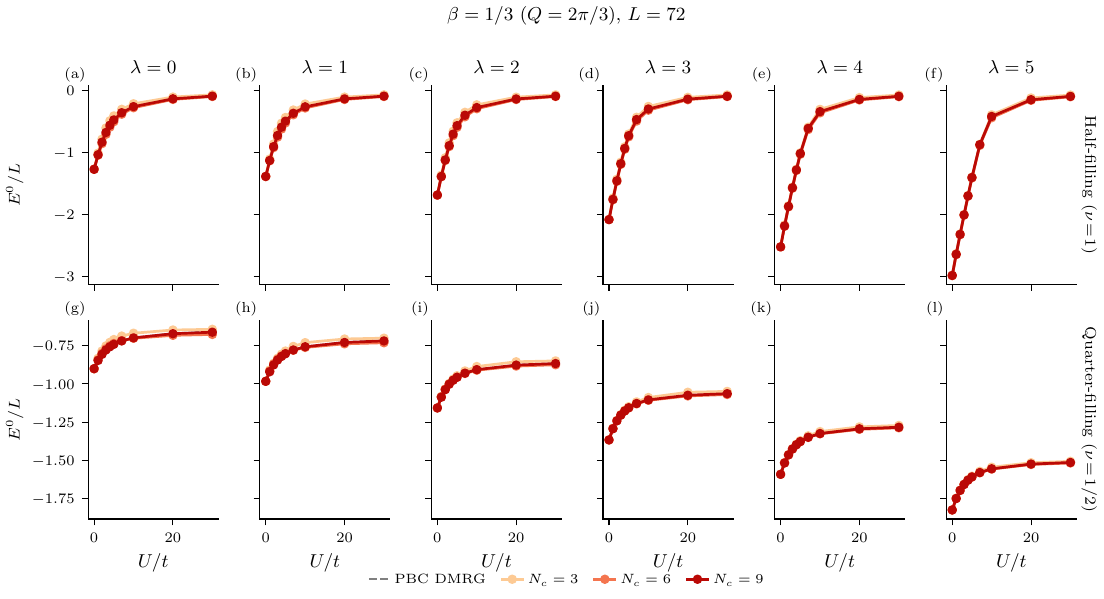}
    \captionof{figure}{Absolute ground state energy per site $E^0/L$ as a function of $U/t$ for $\beta=1/3$ ($Q=2\pi/3$), $L=72$. Columns correspond to different $\lambda$ values; lines show different cluster sizes $N_c$; grey dashed is PBC DMRG ($\chi=256$).}
    \label{fig:v_abs_energy_1-3}
\end{center}

\begin{center}
    \includegraphics[width=0.98\linewidth]{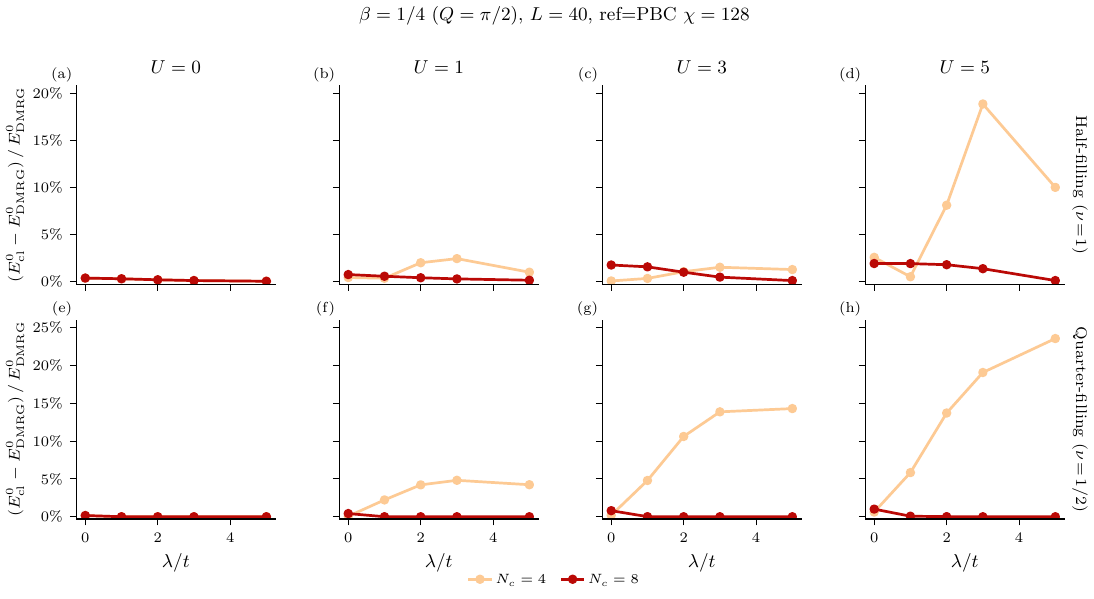}
    \captionof{figure}{Relative energy error as a function of $\lambda/t$ at $\beta=1/4$ ($Q=\pi/2$), $L=40$. Top row: half-filling. Bottom row: quarter-filling. Reference is PBC DMRG ($\chi=128$).}
    \label{fig:scheme_comparison_finite_v_t_1pi4}
\end{center}

\begin{center}
    \includegraphics[width=0.98\linewidth]{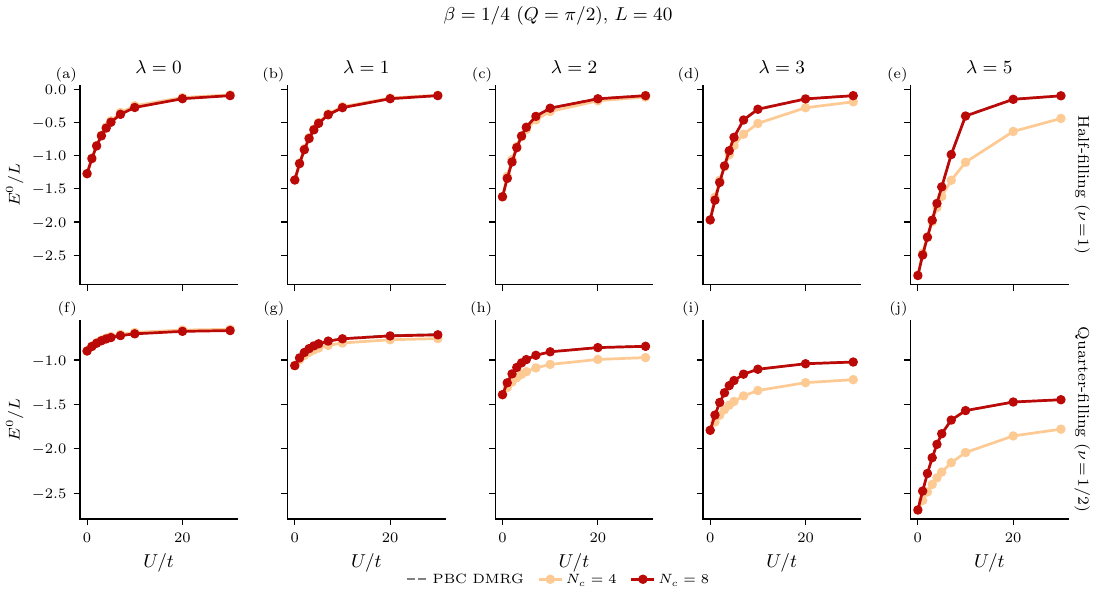}
    \captionof{figure}{Absolute ground state energy per site $E^0/L$ as a function of $U/t$ for $\beta=1/4$ ($Q=\pi/2$), $L=40$. Columns correspond to different $\lambda$ values; lines show different cluster sizes $N_c$; grey dashed is PBC DMRG ($\chi=128$).}
    \label{fig:v_abs_energy_1-4}
\end{center}

\Cref{fig:v_compressibility} shows the filling curves for the AAH model at $\beta=1/2$ with finite $\lambda$, comparing different cluster sizes $N_c$. The quasiperiodic potential modifies the filling structure, particularly at large $\lambda$ where it competes with the Hubbard $U$.

\begin{center}
    \includegraphics[width=0.98\linewidth]{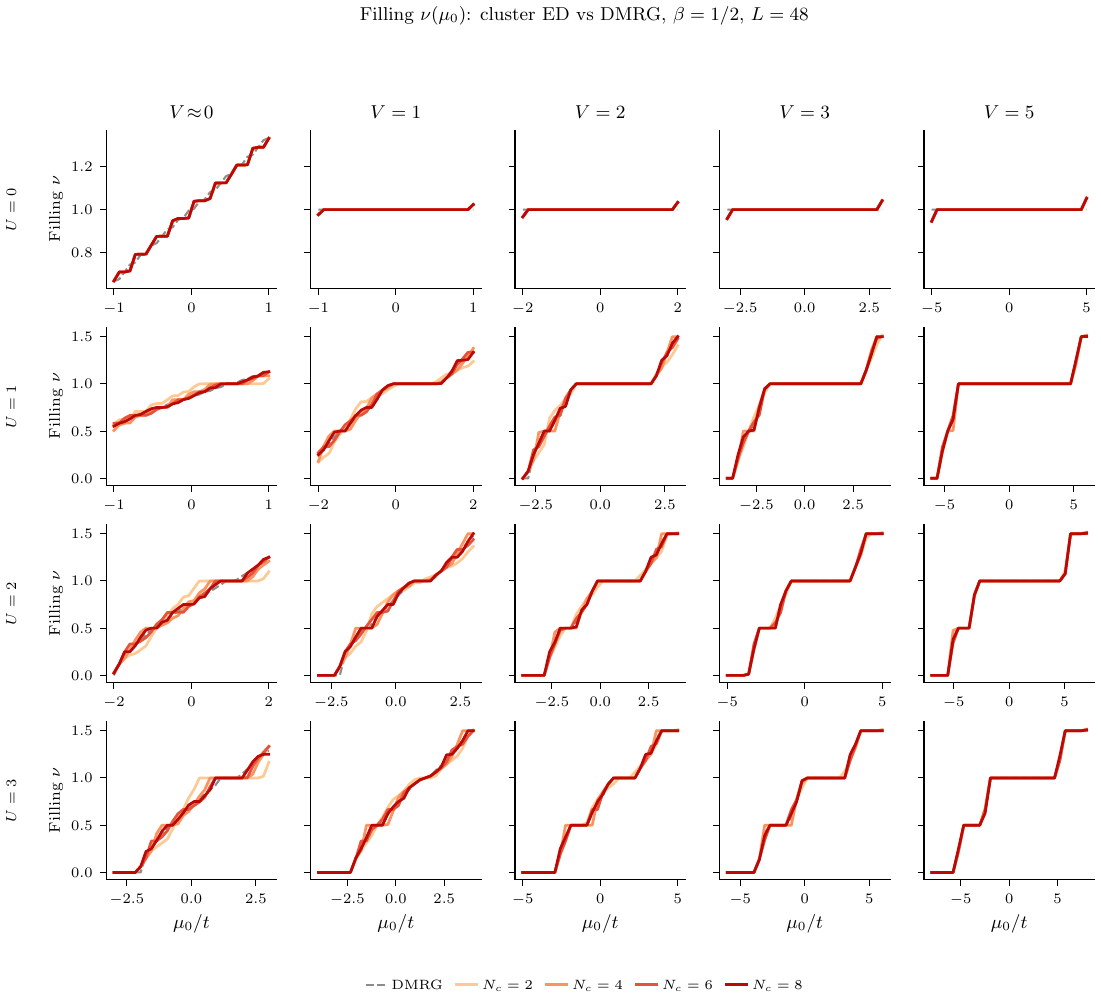}
    \captionof{figure}{Filling $\nu(\mu_0)$ for $\beta=1/2$, $L=48$, at various $U$ and $\lambda$ values. Rows correspond to $U$ values; columns to $\lambda$ values. Different cluster sizes $N_c$ are shown as colored lines; grey dashed is OBC DMRG ($\chi=32$).}
    \label{fig:v_compressibility}
\end{center}

\Cref{fig:v_filling_relerr} shows the relative filling error. As with the pure Hubbard case, the largest errors occur where the filling changes rapidly with $\mu_0$.

\begin{center}
    \includegraphics[width=0.98\linewidth]{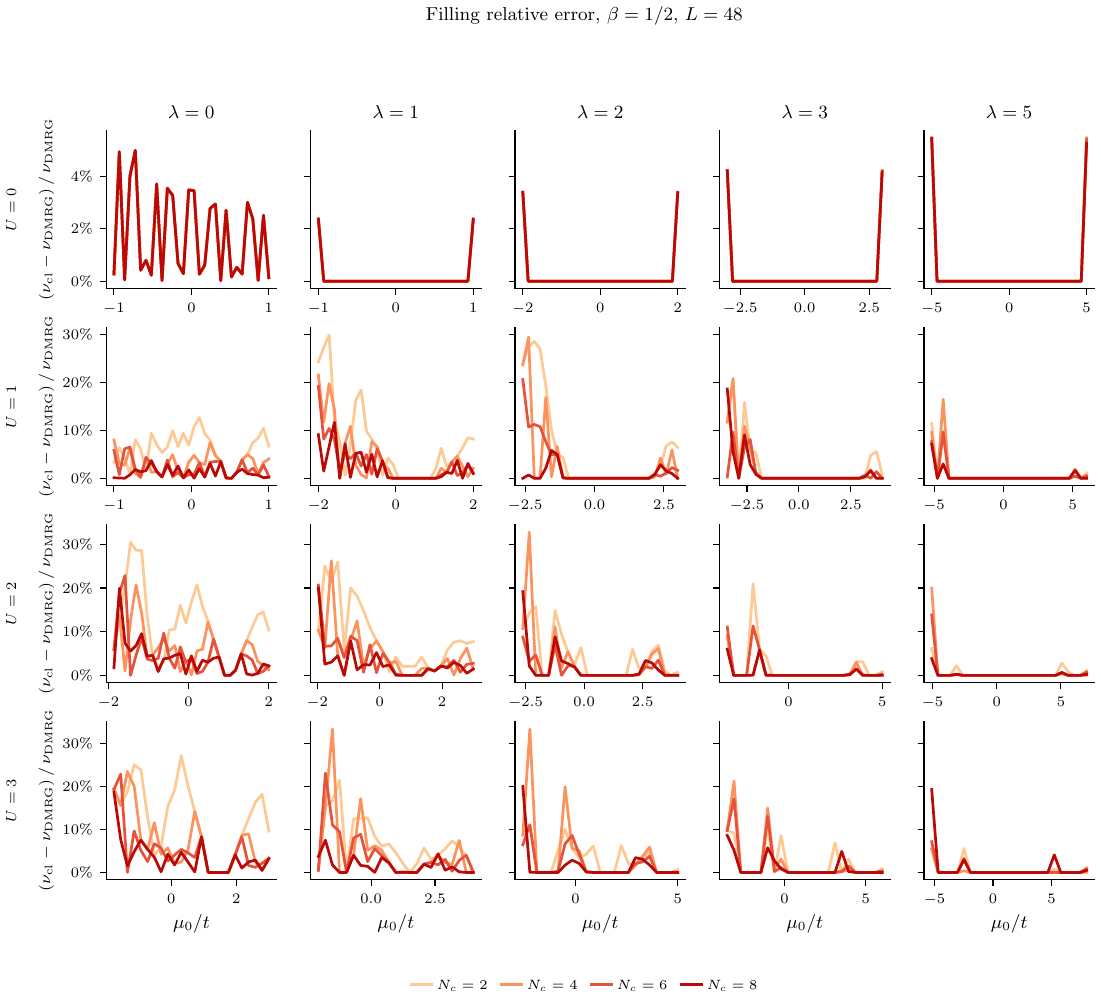}
    \captionof{figure}{Relative filling error $|\nu_{\mathrm{cl}} - \nu_{\mathrm{DMRG}}|/\nu_{\mathrm{DMRG}}$ vs $\mu_0/t$ for $\beta=1/2$, $L=48$. Same layout as \cref{fig:v_compressibility}. Reference is OBC DMRG ($\chi=32$).}
    \label{fig:v_filling_relerr}
\end{center}

\subsection{Additional data for fixed supercluster comparison}
\label{sec:app_fixedsc}

\Cref{fig:fixed_sc_abs_energy} shows the absolute energy for the fixed supercluster comparison at $\beta=1/2$, $L=48$, corresponding to the relative error in \cref{fig:v_convergence_fixedsc_v_pi}. Here the $U$-dependence is shown at each $\lambda$ value for the largest supercluster size, comparing different $(N_c, \Delta)$ pairs.

\begin{center}
    \includegraphics[width=0.98\linewidth]{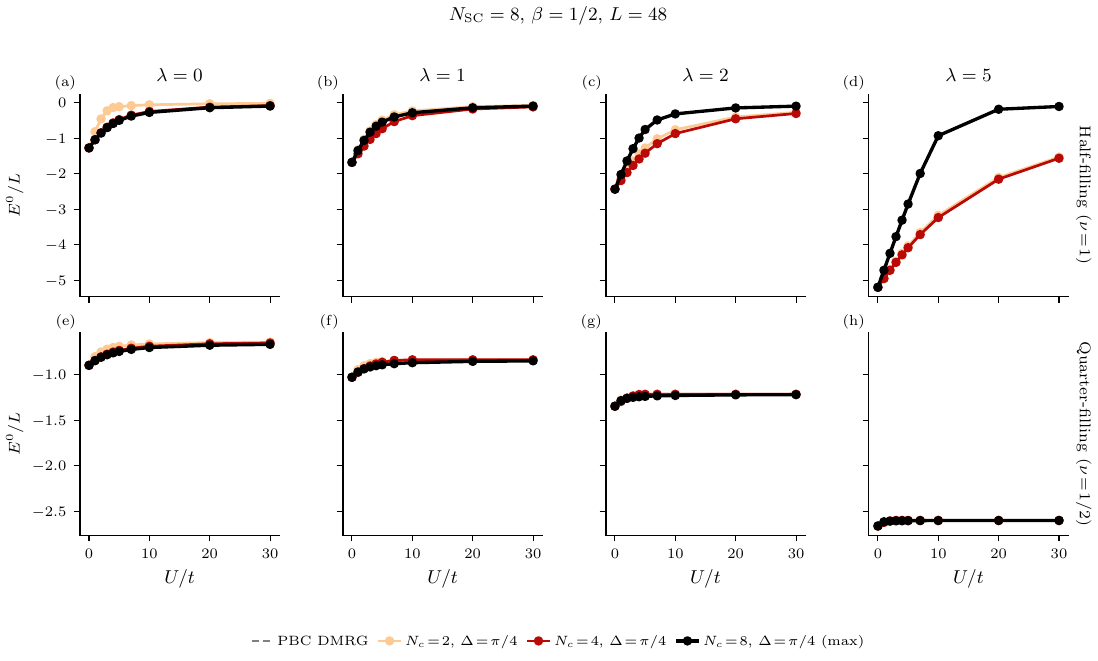}
    \captionof{figure}{Absolute ground state energy per site $E^0/L$ as a function of $U/t$ for fixed supercluster size $N_{\mathrm{SC}}=8$, $\beta=1/2$, $L=48$. Columns correspond to different $\lambda$ values; top row is half-filling, bottom row is quarter-filling. Maximal separation is shown in black; grey dashed is PBC DMRG ($\chi=128$).}
    \label{fig:fixed_sc_abs_energy}
\end{center}

\section{Form of the interaction in the discarded scheme}
\label{sec:appendix_discard_scheme}

In this appendix we derive a microscopic real-space representation of the
cluster-truncated Hubbard interaction under the \emph{discard} (open-boundary)
convention introduced below \cref{eq:cluster_bc}.
Crucially, the discarded scheme remains \emph{block diagonal} in the cluster
label \(K\): clusters do not couple, and the interaction is a sum of independent
cluster interactions.

For notational clarity we work in one dimension with a contiguous index set.
The extension to higher dimensions with rectangular index sets factorizes
componentwise.

\subsection{1D setup: disjoint clusters and the discard rule}

Let the microscopic lattice have \(L\) sites with spacing \(a\), and momenta
\(k = 2\pi j/(La)\), \(j\in\mathbb{Z}_L\).
Fix a cluster size \(N_c\) and a cluster spacing \(\Delta\) (a multiple of
\(2\pi/(La)\)) such that the Brillouin zone decomposes into disjoint clusters
\[
\mathrm{BZ}=\bigsqcup_{K\in\mathcal{K}} \mathcal{C}_{K},
\qquad
\mathcal{C}_{K}:=\{K+k_j:\ j=0,1,\dots,N_c-1\},
\qquad
k_j:=j\Delta,
\]
with \(|\mathcal{K}|=L/N_c\).
(Equivalently, \(\mathcal{K}\) is a set of coset representatives and
\(\mathcal{C}=\{k_j\}\) is the fixed set of relative momenta.)

In the \emph{discard} convention, for a given transfer \(q=k_m=m\Delta\) we keep
only those terms for which \(k_j+q\) and \(k_{j'}-q\) remain inside the index set
\(\{0,\dots,N_c-1\}\).  For minimally separated elements this means
\[
j+m\in\{0,\dots,N_c-1\},\qquad j'-m\in\{0,\dots,N_c-1\}.
\]

\subsection{Cluster interaction in momentum space (manifest \(K\)-decoupling)}

Define the truncated (discarded) coarse cluster density operator (CCDO)
\begin{equation}
\label{eq:rho_discard_def}
\rho^{(K)}_{\sigma}(m)
:=\sum_{\substack{j=0\\ j+m\in[0,N_c-1]}}^{N_c-1}
c^{\dagger}_{K+k_{j+m},\sigma}\,c_{K+k_j,\sigma},
\qquad m\in\mathbb{Z},\ |m|\le N_c-1.
\end{equation}
(For \(m<0\) this is the same definition, i.e. \(j+m\) must still lie in
\([0,N_c-1]\).)

Then the discarded cluster interaction is
\begin{equation}
\label{eq:Hint_discard_CCDO}
H^{\mathrm{(disc)}}_{\mathrm{int}}
=\frac{U}{N_c}\sum_{K\in\mathcal{K}}\;
\sum_{m=-(N_c-1)}^{N_c-1}
\rho^{(K)}_{\uparrow}(m)\,
\rho^{(K)}_{\downarrow}(-m).
\end{equation}
This expression is, as in the wrap scheme, exactly block diagonal in \(K\).

\subsection{Microscopic real-space form and the appearance of triangular weights}

We now Fourier transform Eq.~\eqref{eq:Hint_discard_CCDO} to the microscopic
lattice.  Substituting
\(
c_{p,\sigma}=\frac{1}{\sqrt{L}}\sum_{R}e^{-ipR}c_{R,\sigma}
\),
into the CCDO, we obtain:
\begin{equation}
\label{eq:rho_discard_realspace_bilinear}
\rho^{(K)}_{\sigma}(m)
=\frac{1}{L}\sum_{R_1,R_2}
e^{iK(R_1-R_2)}
\Bigg[\sum_{\substack{j=0\\ j+m\in[0,N_c-1]}}^{N_c-1}
e^{ik_j(R_1-R_2)}\Bigg]
e^{i k_m R_1}\;
c^{\dagger}_{R_1,\sigma}c_{R_2,\sigma}.
\end{equation}
The bracket is a \emph{truncated} Dirichlet sum because of the discard rule.
For the contiguous set \(k_j=j\Delta\), the allowed values of \(j\) are
\(j=0,\dots,N_c-1-m\) when \(m\ge 0\), and \(j=-m,\dots,N_c-1\) when \(m<0\).
Define the truncated Dirichlet kernel
\begin{equation}
\label{eq:Dirichlet_trunc_def}
\mathcal{D}_{M}(x)
:=\sum_{j=0}^{M-1}e^{ij\Delta x}
=
e^{i\frac{(M-1)\Delta x}{2}}
\frac{\sin\!\big(\frac{M\Delta x}{2}\big)}{\sin\!\big(\frac{\Delta x}{2}\big)}.
\end{equation}
Then the bracket in Eq.~\eqref{eq:rho_discard_realspace_bilinear} can be written
as
\begin{equation}
\label{eq:bracket_as_truncD}
\sum_{\substack{j=0\\ j+m\in[0,N_c-1]}}^{N_c-1}
e^{ik_j(R_1-R_2)}
=
\begin{cases}
\mathcal{D}_{N_c-m}(R_1-R_2), & m\ge 0,\\[4pt]
e^{-i m\Delta (R_1-R_2)}\mathcal{D}_{N_c-|m|}(R_1-R_2), & m<0,
\end{cases}
\end{equation}
where the extra phase for \(m<0\) arises from the shifted summation window \(j=|m|,\dots,N_c-1\) and is absorbed into the overall phase in Eq.~\eqref{eq:W_discard_def}.

Substituting Eq.~\eqref{eq:rho_discard_realspace_bilinear} into
Eq.~\eqref{eq:Hint_discard_CCDO}, one obtains after a short calculation
\begin{align}
\label{eq:Hint_discard_realspace_general}
H^{\mathrm{(disc)}}_{\mathrm{int}}
&=\frac{U}{N_c\,L^2}\sum_{K\in\mathcal{K}}
\sum_{R_1,R_2,R_3,R_4}
e^{iK(R_1-R_2+R_3-R_4)}\;
\mathcal{W}(R_1-R_2,R_1-R_3,R_3-R_4)\nonumber\\
&\qquad\times
c^{\dagger}_{R_1,\uparrow}c_{R_2,\uparrow}\,
c^{\dagger}_{R_3,\downarrow}c_{R_4,\downarrow},
\end{align}
where all dependence on the discard boundary condition is packaged into the
weight
\begin{equation}
\label{eq:W_discard_def}
\mathcal{W}(x,y,z)
:=
\sum_{m=-(N_c-1)}^{N_c-1}
\Big[\mathcal{D}_{N_c-|m|}(x)\Big]\,
\Big[\mathcal{D}_{N_c-|m|}(z)\Big]\,
e^{i m\Delta\, y},
\end{equation}
with \(x=R_1-R_2\), \(y=R_1-R_3\), \(z=R_3-R_4\).
Equation~\eqref{eq:W_discard_def} highlights the key distinction from the wrap
case: the discard rule causes the Dirichlet kernels to shrink with increasing
momentum transfer, since the number of terms retained in each truncated sum is
\(N_c-|m|\), decreasing linearly from \(N_c\) at \(m=0\) to \(1\) at
\(|m|=N_c-1\). This \(m\)-dependent truncation couples the three spatial
arguments of \(\mathcal{W}\) and prevents the factorization that, in the wrap
case, follows from the discrete convolution theorem. The only Fourier dependence
on the cross-spin separation enters as a single oscillatory factor
\(e^{im\Delta (R_1-R_3)}\).

Equation~\eqref{eq:Hint_discard_realspace_general} is the microscopic real-space
form of the discarded cluster interaction.  The \(\sum_{K\in\mathcal K}\) factor
enforces center-of-mass conservation only \emph{modulo the coarse lattice}
dual to \(\mathcal K\); correspondingly, the interaction is not strictly local on
the microscopic lattice unless \(N_c=L\).

The primary algebraic distinction from the wrap convention is the appearance of
the triangular factor \(N_c-|m|\), which damps the oscillatory Dirichlet-type sums
and yields a nonnegative, smoothly decaying real-space interaction.

\twocolumngrid
\bibliography{refs}

\end{document}